\documentclass[journal]{IEEEtran}

\usepackage{multicol}
\usepackage{multirow}
\usepackage{longtable}
\usepackage{subcaption}

\usepackage{balance}
\usepackage{cite}

\usepackage{algorithm}
\usepackage{algpseudocode}

\ifCLASSINFOpdf
   \usepackage[pdftex]{graphicx}
   \DeclareGraphicsExtensions{.pdf,.jpeg,.png}
\else
 
\fi

\usepackage{mathrsfs,amsmath}
\usepackage{amssymb}

\usepackage{amsthm,hyperref,enumerate}
\newtheorem{theorem}{Theorem}

\newtheorem{problem}{Problem}

\theoremstyle{definition}
\newtheorem{defn}{Definition}

\theoremstyle{remark}
\newtheorem{remark}{Remark}
\newtheorem{lemma}{Lemma}

\theoremstyle{plain}
 \newtheorem{assumption}{Assumption}

\usepackage{ragged2e}

\usepackage[amssymb]{SIunits}
\hypersetup{
	colorlinks=true,
	linkcolor=blue,
	filecolor=blue,      
	urlcolor=blue,
}

\begin{document}

\title{3-D Distributed Localization with Mixed Local Relative Measurements}

\author{Xu Fang, Xiaolei Li, Lihua Xie,~\IEEEmembership{Fellow,~IEEE}
\thanks{This work is partially supported by National Research Foundation (NRF) Singapore, ST Engineering-NTU Corporate Lab under its NRF Corporate Lab @ University Scheme and National Natural Science Foundation of China (NSFC) under Grant 61720106011 and 61903319. (Corresponding author: Lihua Xie.)}
\thanks{The authors are with the School of Electrical and Electronic Engineering, Nanyang Technological University, Singapore. (E-mail: fa0001xu@e.ntu.edu.sg; xiaolei@ntu.edu.sg;
elhxie@ntu.edu.sg).}
}

\maketitle

\begin{abstract}
This paper studies 3-D distributed network localization using mixed types of local relative measurements. 
Each node holds a local coordinate frame without a common orientation and can only measure one type of information (relative position, distance, relative bearing, angle, or ratio-of-distance measurements) about its neighboring nodes in its local coordinate frame. A novel rigidity-theory-based distributed localization is developed to overcome the challenge due to the absence of a global coordinate frame. The main idea is to construct displacement constraints for the positions of the nodes by using mixed local relative measurements. Then, a linear distributed localization algorithm is proposed for each free node to estimate its position by solving the displacement constraints. The algebraic condition and
graph condition are obtained 
to guarantee the global convergence of the proposed distributed localization algorithm.
\end{abstract}


\IEEEpeerreviewmaketitle

\section{Introduction}

\IEEEPARstart{N}{etwork} localization is a fundamental problem in multiagent-related applications such as formation control, cooperative pursuit, target tracking, etc \cite{jiang2019group, li2020adaptive,wei2019vision,fan2019zoning}. There are typically
two kinds of nodes in a network: anchor nodes and free nodes. The research on network localization focuses on how to localize the unknown free nodes by the known anchor nodes and relative measurements.

In large-scale networks, it may not be practical to equip each sensor with a GPS due to cost and in some environments GPS is not available \cite{chen2016smartphone}. In addition, there is usually no central unit that can obtain all the measurements in the network and compute all the positions of the free nodes in a centralized way. Hence, it is more practical to design a distributed localization algorithm such that each free node can estimate its own position by only communicating with its neighboring nodes. The existing distributed
algorithms can be divided into
four classes:  distance-based \cite{ diao2014barycentric, han2017barycentric}, relative-position-based \cite{barooah2007estimation}, angle-based \cite{jing2019angle1}, and bearing-based \cite{zhao2016localizability, piovan2013frame, shames2012analysis}. 

The distance-based network localization has been studied extensively so far. A distance-based network is localizable if and only if it is globally rigid. Based on a generalized barycentric coordinate representation,  distance-based distributed protocols are proposed to achieve network localization in both two-dimensional and three-dimensional spaces \cite{diao2014barycentric, han2017barycentric}. Similar result can also be found in the relative-position-based network localization \cite{barooah2007estimation}. When each node can measure angles,
a semi-definite strategy is used to localize the free nodes iteratively \cite{jing2019angle1}. 
In the bearing-based network localization, Zhao has proved that an infinitesimally bearing rigid network is localizable \cite{zhao2016localizability}. 
In addition, the problem of how to make use of ratio-of-distance measurements
to localize the free nodes is still not solved.

Note that the aforementioned works assume that the measurements of all the nodes are of identical type.  In a large-scale network, different node may be equipped with different sensor and has different sensing capability.
It is more practical to consider 
the network localization with mixed measurements. Some nodes may measure only relative distances, while others may measure only relative bearing, angle, relative position, or ratio-of-distance.  
The existing works consider distributed localization with mixed distance and bearing measurements \cite{lin2015distributed, eren2011cooperative, stacey2017role}, but \cite{lin2015distributed} is limited to 2-D space and \cite{eren2011cooperative, stacey2017role} need a global coordinate frame. Another work \cite{lin2017mix} develops a distributed algorithm for 2-D network localization with mixed distance, bearing, and relative position measurements. Note that the 3-D case cannot be solved by trivially extending the results in \cite{lin2015distributed, lin2017mix}.

Different from the works \cite{lin2015distributed,eren2011cooperative, lin2017mix, stacey2017role}, we consider 3-D
network localization with mixed local relative measurements without any known global coordinate frame. Inspired by our recent work 
that studies 3-D
network localization with the same type of local relative measurements \cite{fang2020angle}, in this paper, we study the network localization with mixed types of measurements, which is more challenging.
Each node
can only have one of the five types of measurements (relative position, distance, relative bearing, angle, and ratio-of-distance measurements) about its neighboring nodes in its local coordinate frame.
A novel rigidity-theory-based distributed localization is developed to overcome the challenge due to the absence of a global coordinate frame. The key idea is to construct displacement constraints for the positions of the nodes by using mixed local relative measurements. Then, a linear distributed localization algorithm is proposed for each free node to estimate its own position by 
solving the displacement constraints. Moreover, algebraic condition and
graph condition are obtained 
to ensure global convergence of the distributed localization algorithm. 

The remaining parts of the paper are structured as follows: angle-displacement rigidity theory and problem statement are given in Section \ref{preli} and Section \ref{statement}, respectively. The localization problem with mixed local relative measurements is formulated in Section \ref{fomu}, where the corresponding algebraic and graph conditions for localizability are given. In addition, a distributed localization algorithm is proposed to ensure global convergence of the distributed localization algorithm.
Some numerical examples are given in Section \ref{simu} to illustrate the theoretical results. Section \ref{conc1} ends this paper with some conclusions.

\section{Preliminaries of Angle-displacement Rigidity Theory}\label{preli}

\subsection{Notations}

Let $\| \cdot \|$, $\text{Null}(\cdot)$,   $\text{dim}(\cdot)$, $\text{Span}(\cdot)$, $| \cdot|$, and $\text{Rank}(\cdot)$ denote the $L_2$ norm, null space, dimension, span, cardinality rank of a given vector or matrix, respectively.  The Kronecker product and the determinant of a square matrix are denoted by $\otimes$ and $\text{det}(\cdot)$, respectively.
${I}_d \in \mathbb{R}^{d \times d}$ is the identity matrix, and $\mathbf{1}_d \triangleq (1,\cdots,1)^T, \mathbf{0}_d\triangleq(0,\cdots,0)^T \in \mathbb{R}^d$. An undirected graph $\mathcal{G}=\{ \mathcal{V},\mathcal{E}\}$ consists of a vertex set $\mathcal{V}$ of elements called nodes and an edge set $\mathcal{E}  \subseteq \mathcal{V} \times \mathcal{V}$ of ordered pairs of nodes called edges, where $(i,j) \in \mathcal{E} \Leftrightarrow (j,i) \in \mathcal{E}$. The set of neighbors of node $i$ is denoted as $\mathcal{N}_i := \{ j \in \mathcal{V} : (i,j) \in \mathcal{E} \}$. Suppose $\Sigma_g$ is a global coordinate frame in $\mathbb{R}^3$.  Let $p_i, p_j \in \mathbb{R}^3$ be the positions of nodes $i, j$ in $\Sigma_g$, and $e_{ij} \triangleq p_j \!- \! p_i$ is the relative position in $\Sigma_g$.

\begin{figure}[t]
\centering
\includegraphics[width=0.65\linewidth]{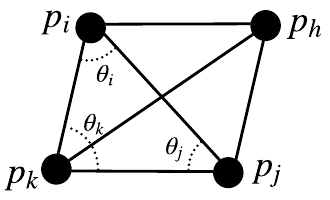}
\caption{Angle and displacement constraints based network. }
\label{moti2}
\end{figure}

Next, we introduce the concepts of angle constraint and displacement constraint, which will be used for 3-D network localization with mixed types of measurements.

\subsection{Angle Constraint and Displacement Constraint}\label{constraint}

\begin{lemma}\label{angle12}

\cite{fangtechnique} For any three different nodes $i,j,k$
in $\mathbb{R}^3$, 
denote $\theta_i, \theta_j, \theta_k \in [0, \pi]$ as the angles between $e_{ij}$ and $e_{ik}$, $e_{ji}$ and $e_{jk}$,  $e_{ki}$ and $e_{kj}$ respectively.  The angles $\theta_i, \theta_j, \theta_k \in [0, \pi]$ are determined uniquely by the parameters $w_{ik}, w_{ki}, w_{ij}, w_{ji}, w_{jk}, w_{kj}$ in \eqref{ang1}-\eqref{aa2}.
\begin{align}\label{ang1} 
 &w_{ik}e_{ik}^Te_{ij}+w_{ki}e_{ki}^Te_{kj}=0, \\
  \label{aa1}
 &w_{ij}e_{ij}^Te_{ik}+w_{ji}e_{ji}^Te_{jk}=0, \\
 \label{aa2}
 &w_{jk}e_{jk}^Te_{ji}+w_{kj}e_{kj}^Te_{ki}=0,
\end{align}
where $w_{ik}^2+w_{ki}^2 \neq 0$, $w_{ij}^2+w_{ji}^2 \neq 0$, and $w_{jk}^2+w_{kj}^2 \neq 0$. 
\end{lemma}

Note that in 
\textit{Lemma \ref{angle12}}, $\theta_i, \theta_j, \theta_k \in [0, \pi]$ means that
$p_i, p_j, p_k$ can be either colinear or non-colinear. 

\begin{defn}\label{d1}
 Equations \eqref{ang1},  \eqref{aa1}, and \eqref{aa2} are formally defined as angle constraints for the nodes $i, j, k$  with respect to the edges $(i,j), (i,k), (j,k)$. 
\end{defn}

A
wide matrix $E_i=(e_{ij}, e_{ik}, e_{ih}, e_{il})$ $ \in \mathbb{R}^{3 \times 4}$ can be constructed based on five different nodes $i,j,k,h,l$. From the matrix theory, for any wide matrix $E_i$, there must be a non-zero vector $ \mu_i=(\mu_{ij}, \mu_{ik}, \mu_{ih}, \mu_{il})^T \in \mathbb{R}^4$ satisfying $E_i\mu_i= \mathbf{0}$, i.e.,

\begin{equation}\label{root}
     \mu_{ij}e_{ij}+\mu_{ik}e_{ik}+\mu_{ih}e_{ih} + \mu_{il}e_{il}= \mathbf{0},
\end{equation}
where $\mu_{ij}^2+\mu_{ik}^2+\mu_{ih}^2+\mu_{il}^2 \neq 0$.

\begin{defn}\label{d2}
Equation \eqref{root} is formally defined as a displacement constraint for node $i$ with respect to the edges $(i,j), (i,k), (i,h), (i,l)$.
\end{defn}

It is worth noting that the angle constraints \eqref{ang1}-\eqref{aa2} are constructed by relative positions rather than angles. 
Note that the anchor positions are known, that is, the relative positions among the anchors are known. Hence, the angle constraints among the anchors are available.
The parameters $w_{ik}$, $w_{ki}$, $w_{ij}$, $w_{ji}$, $w_{jk}$, $w_{kj}$ in \eqref{ang1}-\eqref{aa2} are obtained by the relative positions $e_{ij}, e_{ik}, e_{jk}$. For example,
the parameters $w_{ik}$ and $w_{ki}$ in \eqref{ang1} can be designed as
\begin{equation}\label{para}
\begin{array}{ll}
&  \! \begin{array}{lll} w_{ik} = \frac{1}{e_{ik}^Te_{ij}},  w_{ki} = \frac{-1}{e_{ki}^Te_{kj}}, 
& e_{ik}^Te_{ij}, e_{ki}^Te_{kj} \neq 0, \\ 
w_{ik} = 1,  w_{ki} = 0, & 
e_{ik}^Te_{ij}=0, e_{ki}^Te_{kj} \neq 0, \\ 
w_{ik} = 0,  w_{ki} = 1, & 
e_{ik}^Te_{ij} \neq 0,   e_{ki}^Te_{kj} = 0.\\
\end{array}
\end{array} \\
\end{equation}

\begin{remark}
Since the nodes $i,j,k$ are not collocated, the case $e_{ik}^Te_{ij} = 0,   e_{ki}^Te_{kj} = 0$ does not hold. Without loss of generality, suppose $e_{ik}^Te_{ij} \neq 0$. 
The parameters $w_{ik}, w_{ki}$ satisfying the angle constraint $w_{ik}e_{ik}^Te_{ij}+w_{ki}e_{ki}^Te_{kj}= 0$ are not unique, but the ratio of parameters $\frac{w_{ik}}{w_{ki}}=-\frac{e_{ki}^Te_{kj}}{e_{ik}^Te_{ij}}$ satisfying the angle constraint $w_{ik}e_{ik}^Te_{ij}+w_{ki}e_{ki}^Te_{kj}= 0$ is unique. The angles  $\theta_i, \theta_j, \theta_k$ between the anchor nodes $i,j,k$ are determined uniquely by these unique ratios of parameters in the angle constraints \eqref{ang1}-\eqref{aa2}.
\end{remark}

In this paper, the known anchor positions will be used to construct the angle constraints, which do not need any measurement. In the proposed network localization, the mixed types of measurements are only used for constructing the displacement constraints \eqref{root}. The details of how to make use of local relative measurements to construct the displacement constraint will be introduced in Section \ref{conc}.
Note that network localizability is used to specify whether the corresponding network structure is unique. The following angle-displacement rigidity theory will be used 
to explore the network localizability.

\subsection{Angle-displacement Rigidity Theory}

A three-dimensional network, represented by
$(\mathcal{G}, p)$, is an undirected graph $\mathcal{G}$ with its vertex $i$ mapped to point $p_i$, where $p=(p_1^T, \cdots, p_{n}^T)^T \in \mathbb{R}^{3n}$ is the configuration in $\mathcal{G}$.
The structure of a network can be determined by angle and displacement constraints. A simple example is given in Fig. \ref{moti2}. The structure of the network $\diamondsuit_{ijkh}(p)$ is determined by three angles $\theta_i, \theta_j, \theta_k$ and a displacement constraint $e_{hi}+e_{hj}-e_{hk}=0$, where $\theta_i, \theta_j, \theta_k$ are determined by three angle constraints
\eqref{ang1}-\eqref{aa2}.
Thus, the structure of the network $\diamondsuit_{ijkh}(p)$ can be decided by three angle constraints \eqref{ang1}-\eqref{aa2} and the displacement constraint $e_{hi}+e_{hj}-e_{hk}=0$.
The following angle-displacement rigidity theory answers that the structure of a network in $\mathbb{R}^3$ can be uniquely determined
by
a set of angle and displacement constraints up to directly similar transformations, i.e., translation, rotation, and scaling.

\subsubsection{Angle-displacement Function} Let $\Upsilon_{\mathcal{G}}=\{ ( i, j, k) \in \mathcal{V}^{3} : (i,j), (i,k), (j,k) \in \mathcal{E}, i< \! j \!<\! k \}$ with $|\Upsilon_{\mathcal{G}}|=m_r$. Based on Definition \ref{d1}, the angle function $B_{\Upsilon_{\mathcal{G}}}(p) : \mathbb{R}^{3n} \rightarrow \mathbb{R}^{m_r}$ is defined as
\begin{equation}\label{funa}
B_{\Upsilon_{\mathcal{G}}}(p)
\!=\!  (\cdots, w_{ik}e_{ik}^Te_{ij}+w_{ki}e_{ki}^Te_{kj}, \cdots)^T, 
\end{equation}
where $(i,j,k) \! \in \! \Upsilon_{\mathcal{G}}$, and $w_{ik}e_{ik}^Te_{ij}+w_{ki}e_{ki}^Te_{kj}=0$ with $w_{ik}^2+w_{ki}^2 \neq 0$. 

Let 
$\mathcal{X}_{\mathcal{G}}=\{ ( i, j, k, h, l) \in \mathcal{V}^{5} : (i,j), (i,k), $ $ (i,h),  (i,l)  \in \mathcal{E},   j \!<\! k \!<\! h \!<\! l\}$ with $|\mathcal{X}_{\mathcal{G}}|=m_d$. Based on Definition \ref{d2},
the displacement function $L_{\mathcal{X}_{\mathcal{G}}}(p) : \mathbb{R}^{3n} \rightarrow \mathbb{R}^{3m_d}$
is defined as 
\begin{equation}\label{func}
L_{\mathcal{X}_{\mathcal{G}}}(p)
\!=\! (\cdots, \mu_{ij}e_{ij}^T\!+\! \mu_{ik}e_{ik}^T\!+\! \mu_{ih}e_{ih}^T\!+\! \mu_{il}e_{il}^T, \cdots)^T,
\end{equation}
where $( i, j, k, h, l) \in \mathcal{X}_{\mathcal{G}}$, and $\mu_{ij}e_{ij}\!+\! \mu_{ik}e_{ik}\!+\! \mu_{ih}e_{ih}\!+\! \mu_{il}e_{il}=0$ with $\mu_{ij}^2+\mu_{ik}^2+\mu_{ih}^2+\mu_{il}^2 \neq 0$. Let $\mathcal{T}_{\mathcal{G}} = \Upsilon_{\mathcal{G}} \cup \mathcal{X}_{\mathcal{G}} $ and $m=3m_d+m_r$. Then, the angle-displacement function $f_{\mathcal{T}_{\mathcal{G}}}(p) : \mathbb{R}^{3n} \rightarrow \mathbb{R}^{m}$ is defined as
\begin{equation}\label{cust}
    f_{\mathcal{T}_{\mathcal{G}}}(p) = (B_{\Upsilon_{\mathcal{G}}}^T(p), L^T_{\mathcal{X}_{\mathcal{G}}}(p)  )^T.
\end{equation}

\begin{remark}
Note that different configurations may have different angle-displacement function. The customized
angle-displacement function $f_{\mathcal{T}_{\mathcal{G}}}(p)= \mathbf{0}$ for the configuration $p$ may not be applicable to a different configuration $q \neq p$, i.e.,  $f_{\mathcal{T}_{\mathcal{G}}}(q) = \mathbf{0}$ may not hold.
\end{remark}

\subsubsection{Angle-displacement Rigidity Matrix}

The angle-displacement rigidity matrix is defined as
\begin{equation}\label{rigiditym}
R(p) = \frac{\partial f_{\mathcal{T}_{\mathcal{G}}}(p)}{\partial p} \in \mathbb{R}^{m \times 3n}.
\end{equation}

Let $\delta p$ be a variation of the configuration $p$. If 
$R(p)\delta p= 0$, $\delta p$ is called an angle-displacement infinitesimal motion. Since  $ \mu_{ij}e_{ij}\!+\! \mu_{ik}e_{ik}\!+\! \mu_{ih}e_{ih}\!+\! \mu_{il}e_{il}=\mathbf{0}$ and $w_{ik}e_{ik}^Te_{ij}+w_{ki}e_{ki}^Te_{kj}=0$, the angle-displacement infinitesimal motions include translations, rotations, and scalings. 
An angle-displacement infinitesimal motion is called trivial if it corresponds to a translation, a rotation and a scaling of the entire network.

\begin{lemma}\label{inva}
\cite{fangtechnique} The trivial motion space for angle-displacement rigidity in $\mathbb{R}^3$ is $S=S_t \cup S_r \cup S_s$, where $S_t=\{ \mathbf{1}_n\otimes {I}_3 \}$ is the space formed by translations , $S_r= \{(I_n \otimes A)p, A+A^T =\mathbf{0}, A \in \mathbb{R}^{3 \times 3} \}$ is the space formed by rotations, and $S_s= \text{Span}(p)$ is the space formed by scalings.
\end{lemma}

Let 
\begin{equation}
    \begin{array}{ll}
         &  \eta_d^T = \frac{\partial (\mu_{ij}e_{ij}\!+\! \mu_{ik}e_{ik}\!+\! \mu_{ih}e_{ih}\!+\! \mu_{il}e_{il}) }{\partial p} ,\\
         &  \eta_r^T = \frac{\partial (w_{ik}e_{ik}^Te_{ij}+w_{ki}e_{ki}^Te_{kj}) }{\partial p},
    \end{array}
\end{equation}
where $\eta_d^T$ and $\eta_r^T$ are, respectively,  arbitrary rows of $R(p)$ corresponding to displacement constraint and angle constraint shown as
\begin{equation}\label{infid}
 \eta_d\!=\![\mathbf{0}, 	\!-\!\mu_{ij}\!-\!\mu_{ik}\!-\!\mu_{ih}\!-\!\mu_{il}, \mathbf{0}, \mu_{ij}, \mathbf{0}, \mu_{ik}, \mathbf{0}, \mu_{ih}, \mathbf{0}, \mu_{il},  \mathbf{0}]^T \otimes I_3.   
\end{equation}

\begin{equation}\label{infi}
	\eta_r \!=\! \left[ \!
	\begin{array}{c}
	\mathbf{0} \\
	2w_{ik}p_i+ (w_{ki}-w_{ik})p_j-(w_{ik}+w_{ki})p_k\\
	\mathbf{0} \\
	(w_{ki}-w_{ik})p_i+(w_{ik}-w_{ki})p_k \\
	\mathbf{0} \\
	-(w_{ik}+w_{ki})p_i+(w_{ik}-w_{ki})p_j+2w_{ki}p_k\\
	\mathbf{0} 
	\end{array}
	\right].
\end{equation}

Based on \eqref{infid} and \eqref{infi}, we have
\begin{equation}
    R(p)p = (2B_{\Upsilon_{\mathcal{G}}}^T(p), L^T_{\chi_{\mathcal{G}}}(p)  )^T \in \mathbb{R}^m.
\end{equation}
Then, we obtain
\begin{equation}\label{anre}
p^TR(p)^TR(p)p = 4\| B_{\Upsilon_{\mathcal{G}}}(p) \|^2 + \|L_{\chi_{\mathcal{G}}}(p)\|^2.
\end{equation}

\begin{defn}\label{d3}
\cite{fang2020angle} A network ($\mathcal{G}, p$) is infinitesimally angle-displacement rigid in $\mathbb{R}^3$ if all the angle-displacement infinitesimal motions are trivial. 
\end{defn}

From \textit{Lemma \ref{inva}}, we can know that for 
an infinitesimally angle-displacement rigid network ($\mathcal{G}, p$), the condition $\text{dim}(\text{Null}(R(p)))\!=\!7$ holds. Next, 
we will introduce the problem of 
network localization with mixed types of measurements.

\section{Problem Statement}\label{statement}

The objective is to develop a distributed algorithm for localizing the positions of all unknown nodes in a network with mixed types of measurements including local relative position, distance, local relative bearing, angle, and ratio-of-distance measurements.
Some nodes may measure only distances, while others may  measure only local relative bearings, angles, local relative positions, or
ratio-of-distances.
In addition, 
we assume that the orientation differences between the local coordinate frames of the nodes and global coordinate frame are unknown.

Consider a static network ($\mathcal{G}, p$) in $\mathbb{R}^3$.
Suppose the locations of $n_a$ anchor nodes are given and the locations of $n_f$ free nodes are to be estimated $(n_a\!+\!n_f\!=\!n)$. Denote $\mathcal{V}_a=\{ 1,\cdots, n_a \}$, $\mathcal{V}_f\!=\!\{ n_a \!+\! 1,  \cdots, n \}$, and $\mathcal{V} = \mathcal{V}_a \cup \mathcal{V}_f$. Denote $p_a = (p_1^T, \cdots, p_{n_a}^T)^T$, $p_f = (p_{n_a\!+\!1}^T, \cdots, p_{n}^T)^T$, and  $p = p_a \cup p_f$. Denote $\Sigma_i$ as the 
local coordinate frame of node $i$. Let $Q_i \in SO(3)$ be the unknown rotation matrix from $\Sigma_i$ to $\Sigma_g$, where $\Sigma_g$ is the global coordinate frame.
Define
\begin{equation}\label{local}
e_{ij} = Q_i e_{ij}^{i}, \ \ g_{ij}^{i}= \frac{e_{ij}^i}{d_{ij}}, \ \ \theta_{ijk}=g_{ij}^{i}g_{ik}^{i}, \ \ r_{ijk}=\frac{d_{ij}}{d_{ik}}, 
\end{equation}
where $e_{ij}^{i}$ is the local
relative position in $\Sigma_i$. $d_{ij}=\|e_{ij}\|$ is the distance.
$g_{ij}^{i}$ is the local relative bearing. $\theta_{ijk}$ is the angle between $e_{ij}$ and $e_{ik}$. $r_{ijk}$ is the ratio-of-distance between $d_{ij}$ and $d_{ik}$. 
Since there are
five kinds of measurements in a local coordinate frame,
the nodes 
can be divided into five categories shown as
\begin{enumerate}
    \item $i \in \mathcal{D}$: Node $i$ can measure the distance $d_{ij}$ to its neighboring node $j \in \mathcal{N}_i$.
    \item $i \in \mathcal{B}$: Node $i$ can measure the local relative bearing $g_{ij}^{i}$ with respect to its neighboring node $j \in \mathcal{N}_i$.
    \item $i \in \mathcal{R}$: Node $i$ can measure the local relative position $e_{ij}^{i}$ to its neighboring node $j \in \mathcal{N}_i$.
    \item $i \in \mathcal{A}$: Node $i$ can measure the angle $\theta_{ijk}$ between $e_{ij}$ and $e_{ik}$, where $j, k \in \mathcal{N}_i$. 
    \item $i \in \mathcal{R}o\mathcal{D}$: Node $i$ can measure the ratio-of-distance $r_{ijk}$ between $d_{ij}$ and $d_{ik}$, where $j, k \in \mathcal{N}_i$. 
\end{enumerate}

\begin{remark}
The work in \cite{cao2019ratio} introduces the sensors that can obtain the ratio-of-distance measurements.
\end{remark}

Then, the problem of network localization with mixed types of measurements is described below. 

\begin{problem}
Consider a static network ($\mathcal{G}, p$) with an undirected graph $\mathcal{G}=\{ \mathcal{V},\mathcal{E}\}$, where each node can measure only one of the five kinds of local relative measurements 
$V= \mathcal{R} \cup  \mathcal{B} \cup\mathcal{A}  \cup  \mathcal{D} \cup\mathcal{R}o\mathcal{D}$. Given arbitrary initial estimates $\hat p_i(0), i \in \mathcal{V}_f$, develop a distributed localization algorithm for each free node $i$ such that 
\begin{equation}
    \lim\limits_{t \rightarrow \infty} \hat p_i(t) =p_i,
\end{equation}
where ${\hat p}_i$ is the estimate of $p_i$.
\end{problem}

Next, the method for network localization with mixed types of measurements is presented.

\section{Network Localization with Mixed Types of Measurements} \label{fomu}

In this paper, we consider an undirected graph.
Each node can measure only one of the five 
types of local relative measurements and different nodes may have different measurement types
shown in Fig. \ref{case}. Before introducing the main results, the distance-based displacement constraint is introduced.

\subsection{Distance-based Displacement Constraint}\label{disdis}

In this part, we will introduce how to construct the displacement constraint based on the distance measurements. 
For the node $i$ and its neighbors $j,k,h,l$, define distance matrix $D$ as
\begin{equation}\label{distance}
    D = \left[ \!
\begin{array}{c c c c c}
0 & d_{ij}^2 & d_{ik}^2  & d_{ih}^2 &  d_{il}^2 \\
d_{ji}^2 & 0 & d_{jk}^2 & d_{jh}^2 & d_{jl}^2 \\
d_{ki}^2 & d_{kj}^2 & 0 & d_{kh}^2 & d_{kl}^2 \\
d_{hi}^2 & d_{hj}^2 & d_{hk}^2 & 0 & d_{hl}^2 \\
d_{li}^2 & d_{lj}^2 & d_{lk}^2 & d_{lh}^2 & 0
\end{array}
\right].
\end{equation}

The volume $V_{jkhl}$ of the tetrahedron formed by the nodes $p_j,p_k,p_h,p_l$ can be calculated with distance measurements through Cayley-Menger determinant \cite{sippl1986cayley}. That is, 
\begin{equation}\label{cay1}
\begin{array}{ll}
     & 288V_{jkhl}^2 
    = \\
     & \!\text{det} \! \left ( \!
	\begin{array}{ccc}
	2d^2_{jk} &d^2_{jk}\!+\!d^2_{jh}\!-\!d^2_{kh} &d^2_{jk}\!+\!d^2_{jl}\!-\!d^2_{kl}  \\
	d^2_{jk}\!+\!d^2_{jh}\!-\!d^2_{kh} & 2d^2_{jh} & d^2_{jh}\!+\!d^2_{jl}\!-\!d^2_{hl}\\
	d^2_{jk}\!+\!d^2_{jl}\!-\!d^2_{kl} & d^2_{jh}\!+\!d^2_{jl}\!-\!d^2_{hl} & 2d^2_{jl} \\
	\end{array}
	\! \right). 
\end{array}
\end{equation}

\begin{defn}
If a node $i$ with respect to its neighboring nodes $j,k,h,l$ satisfies the following equations
\begin{equation}\label{baryd}
    \begin{array}{ll}
         &  p_i = \mu_{ij}p_{j}+\mu_{ik}p_{k}+\mu_{ih}p_{h} + \mu_{il}p_{l},\\
         & \mu_{ij}+\mu_{ik}+\mu_{ih}+\mu_{il}=1,
    \end{array}
\end{equation}
then $\{\mu_{ij}, \mu_{ik}, \mu_{ih}, \mu_{il} \}$ is called the barycentric coordinate of node $i$ with respect to the nodes $j,k,h,l$.
\end{defn}

For the volume $V_{jkhl}$ \eqref{cay1},
there are two cases:
\begin{enumerate}[(a)]
\item If $V_{jkhl}^2 \neq 0$,  $p_j,p_k,p_h,p_l$ are non-coplanar. The work in \cite{han2017barycentric} provides a way to obtain the barycentric coordinate $\{\mu_{ij}, \mu_{ik}, \mu_{ih}, \mu_{il} \}$ of node $i$ with respect to the nodes $j, k, h, l$ by only using the distance measurements shown in Algorithm \ref{disa}. This algorithm uses the fact that a congruent framework of the subnetwork consisting of the node and its neighbors
has the same barycentric coordinate. Note that \eqref{baryd} can be rewritten as a displacement constraint shown as
\begin{equation}\label{disd}
     \mu_{ij}e_{ij}+\mu_{ik}e_{ik}+\mu_{ih}e_{ih} + \mu_{il}e_{il}= \mathbf{0}.
\end{equation}

\begin{algorithm}
\caption{Distance-based Barycentric Coordinate ($V_{jkhl} \neq 0$) }
\label{disa}
\begin{algorithmic}[1]
\State Available information: Distance measurements among the nodes $p_i,p_j,p_k,p_h,p_l$. Denote  $(\mathcal{\bar G}, \bar p)$ as a subnetwork with $\bar p=(p_i^T,p_j^T,p_k^T,p_h^T,p_l^T)^T$.
\State Based on the distance matrix $D$ \eqref{distance}, obtaining a congruent network $(\mathcal{\bar G}, \bar q) \cong (\mathcal{\bar G}, \bar p)$ with $\bar q=(q_i^T,q_j^T,q_k^T,q_h^T,q_l^T)^T$ by Algorithm $1$ in \cite{han2017barycentric}, where $q_i, q_j, q_k, $ $ q_h, q_l \in \mathbb{R}^3$. 
\State Based on $\bar q$, calculating the barycentric coordinate of node $p_i$ with respect to its neighboring nodes $p_j,p_k,p_h,p_l$ in \eqref{baryd} by 
\begin{equation}\label{al1}
\mu_{ij}=\frac{V_{iklh}}{V_{jklh}}, \mu_{ik}=\frac{V_{jilh}}{V_{jklh}}, \mu_{ih}=\frac{V_{jkih}}{V_{jklh}}, \mu_{il}=\frac{V_{jkli}}{V_{jklh}},
\end{equation}
where
\begin{equation}\label{vol1}
\begin{array}{ll}
& V_{iklh} = \frac{1}{6} \text{det}(\left[ \!
\begin{array}{c c c c}
1 & 1 & 1  & 1  \\
q_i & q_k & q_l &q_h
\end{array}
\right]),  \\
& V_{jilh} = \frac{1}{6} \text{det}(\left[ \!
\begin{array}{c c c c}
1 & 1 & 1  & 1  \\
q_j & q_i & q_l &q_h
\end{array}
\right]), \\
& V_{jkih} = \frac{1}{6} \text{det}(\left[ \!
\begin{array}{c c c c}
1 & 1 & 1  & 1  \\
q_j & q_k & q_i &q_h
\end{array}
\right]), \\
& V_{jkli} = \frac{1}{6} \text{det}(\left[ \!
\begin{array}{c c c c}
1 & 1 & 1  & 1  \\
q_j & q_k & q_l &q_i
\end{array}
\right]), \\
& V_{jklh} = \frac{1}{6} \text{det}(\left[ \!
\begin{array}{c c c c}
1 & 1 & 1  & 1  \\
q_j & q_k & q_l &q_h
\end{array}
\right]). 
\end{array}
\end{equation}
\end{algorithmic}
\end{algorithm}

\item If $V_{jkhl}^2 = 0$, $p_j,p_k,p_h,p_l$ are coplanar. 
The area $S_{jkh}$ of the triangle formed by the nodes $p_j,p_k,p_h$ can also be calculated with distance measurements through Cayley-Menger determinant \cite{sippl1986cayley}. That is,
\begin{equation}\label{cay2}
 16S_{jkh}^2 =    \text{det} \left ( \!
	\begin{array}{ccc}
	2d^2_{jk} &d^2_{jk}\!+\!d^2_{jh}\!-\!d^2_{kh} \\ d^2_{jk}\!+\!d^2_{jh}\!-\!d^2_{kh} & 2d^2_{jh} \\
	\end{array}
	\right). 
\end{equation}
\end{enumerate}

For the area $S_{jkh}$ \eqref{cay2}, there are also two cases: 
\begin{enumerate}[(i)]
\item If $S_{jkh}^2 \neq 0$, 
$p_j,p_k,p_h$ are non-colinear.
The work in \cite{diao2014barycentric} provides a way to obtain the barycentric coordinate $\{\mu_{lj}, \mu_{lk}, \mu_{lh} \}$ of node $l$ with respect to its neighboring nodes $j, k, h$ in \eqref{ds11} by only using the distance measurements shown in Algorithm \ref{disa1}. 

\begin{equation}\label{ds11}
\begin{array}{ll}
     &  p_l = \mu_{lj}p_j + \mu_{lk}p_k + \mu_{lh}p_h,\\
     & \mu_{lj}+\mu_{lk}+\mu_{lh}=1.
\end{array}
\end{equation}

Note that \eqref{ds11} can also be rewritten as a displacement constraint shown as
\begin{equation}\label{disd1}
     \mu_{ij}e_{ij}+\mu_{ik}e_{ik}+\mu_{ih}e_{ih} + \mu_{il}e_{il}= \mathbf{0},
\end{equation}
where $\mu_{ij}\!=\!-\mu_{lj}, \mu_{ik}\!=\! -\mu_{lk}, \mu_{ih}\!=\! -\mu_{lh}, \mu_{il}\!=\!1$.

\begin{algorithm}
\caption{Distance-based Barycentric Coordinate ($V_{jkhl} = 0, S_{jkh} \neq 0$)}
\label{disa1}
\begin{algorithmic}[1]
\State Available information: Distance measurements among the nodes $p_j,p_k,p_h,p_l$. 
\State 
The absolute values of the barycentric coordinate $\mu_{lj}, \mu_{lk}, \mu_{lh}$ are calculated by 
\begin{equation}\label{al2}
\| \mu_{lj} \| = \frac{\|S_{lkh}\|}{\|S_{jkh}\|}, \| \mu_{lk} \| = \frac{\|S_{lhj}\|}{\|S_{jkh}\|}, \| \mu_{lh} \| = \frac{\|S_{ljk}\|}{\|S_{jkh}\|},
\end{equation}
where $\|S_{jkh}\|$ and others can be solved according to \eqref{cay2}. The absolute values $\| \mu_{lj} \|, \| \mu_{lk} \|, \| \mu_{lh} \|$ satisfy
\begin{equation}\label{sig}
\begin{array}{ll}
     &    \mu_{lj} = \sigma_{lj} \| \mu_{lj} \|, \mu_{lk} = \sigma_{lk} \| \mu_{lk} \|, \mu_{lh} = \sigma_{lh} \| \mu_{lh} \|, \\
     &  \sigma_{lj} \| \mu_{lj} \| + \sigma_{lk} \| \mu_{lk} \|+ \sigma_{lh} \| \mu_{lh} \|=1,
\end{array}
\end{equation}
where $ \sigma_{lj},  \sigma_{lk},  \sigma_{lh}$ take values of either $1$ or $-1$. 
\State Calculating $ \sigma_{lj},  \sigma_{lk}, \sigma_{lh}$ by distance measurements through Algorithm $1$ in the Diao' work \cite{diao2014barycentric}.
\end{algorithmic}
\end{algorithm}

\item If $S_{jkh}^2 \!=\! 0$, $p_j,p_k,p_h$ are colinear. There are three cases: $d_{jk}\!+\!d_{kh}\!=\!d_{jh}$; $d_{jk}\!+\!d_{jh}\!=\!d_{kh}$; $d_{jh}\!+\!d_{kh}\!=\!d_{jk}$.
If $d_{jk}\!+\!d_{kh}\!=\!d_{jh}$, we have
\begin{equation}\label{ds121}
    e_{jk} = \frac{d_{jk}}{d_{kh}}e_{kh}.
\end{equation}

Equation \eqref{ds121} can be rewritten as a displacement constraint shown as
\begin{equation}\label{disd2}
     \mu_{ij}e_{ij}+\mu_{ik}e_{ik}+\mu_{ih}e_{ih} + \mu_{il}e_{il}= \mathbf{0},
\end{equation}
where $\mu_{ij}\!=\!-1, \mu_{ik}\!=\!1\!+\!\frac{d_{jk}}{d_{kh}}, \mu_{ih}\!=\! - \frac{d_{jk}}{d_{kh}}, \mu_{il}\!=\! 0$. If $d_{jk}\!+\!d_{jh}\!=\!d_{kh}$ or $d_{jh}\!+\!d_{kh}\!=\!d_{jk}$, the corresponding displacement constraint \eqref{disd2} can be modified accordingly. The procedure of constructing a distance-based displacement constraint among the nodes $p_i,p_j,p_k,p_h,p_l$ in $\mathbb{R}^3$ is given in \text{Algorithm \ref{array-sum}}.
\end{enumerate}

\begin{algorithm}
\caption{Distance-based Displacement Constraint }
\label{array-sum}
\begin{algorithmic}[1]
\State {Available information: Distance measurements among the nodes $p_i,p_j,p_k,p_h,p_l$.}
\State \textbf{If} $p_j,p_k,p_h,p_l$ 
are non-coplanar $V_{jkhl} \neq 0$ \eqref{cay1}, \textbf{do}
\State \ \ \ \ Constructing a displacement constraint by \eqref{disd};
\State \textbf{Else}
\State \ \ \ \  \textbf{If} $p_j,p_k,p_h$ 
are non-colinear $S_{jkh} \neq 0$ \eqref{cay2}, \textbf{do} 
\State \ \ \ \  \ \ \ \ Constructing a displacement constraint by \eqref{disd1};
\State \ \ \ \  \textbf{Else}
\State \ \ \ \  \ \ \ \ Constructing a displacement constraint by \eqref{disd2};
\State \ \ \ \  \textbf{End}
\State   \textbf{End}
\end{algorithmic}
\end{algorithm}

\begin{figure*}[t]
\centering
\includegraphics[width=1\linewidth]{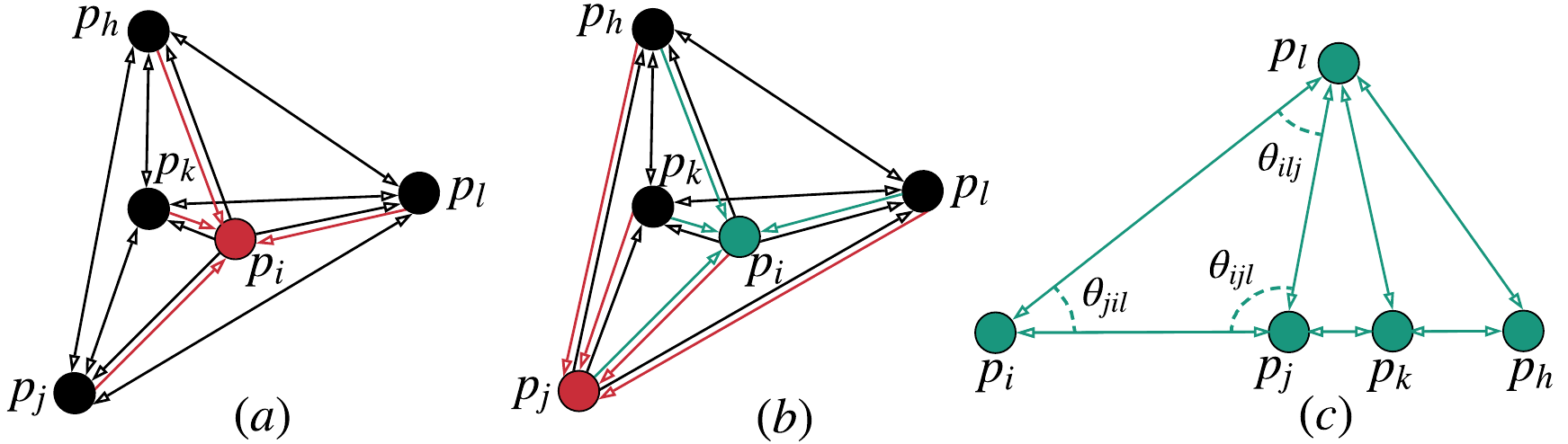}
\caption{The red nodes indicate that they can measure local relative positions. The green nodes mean that they can measure one of the following four types of local relative measurements
(local relative bearing, angle, distance, and ratio-of-distance). The black nodes mean that
they can measure only one of the five types of local relative measurements.   $(a)$ Node $i \in \mathcal{R}$. $(b)$ Node $i \in \mathcal{B} \cup\mathcal{A}  \cup  \mathcal{D} \cup\mathcal{R}o\mathcal{D}$ and Node $j \in \mathcal{R}$. $(c)$ Node $i,j,k,h,l \in  \mathcal{B} \cup\mathcal{A}  \cup  \mathcal{D} \cup\mathcal{R}o\mathcal{D}$.  }
\label{case}
\end{figure*}

\subsection{Construction of Angle Constraints and Displacement Constraints}\label{conc}

\begin{assumption}\label{ass3}
No two nodes are collocated in $\mathbb{R}^3$. Each anchor node $i$ has at least two neighboring anchor nodes $j,k$, where $i,j,k$ are mutual neighbors. Each free node $i$ has at least four neighboring nodes $j,k,h,l$, where $i,j,k,h,l$ are mutual neighbors and non-colinear.
\end{assumption}

Note that network localizability is used to specify whether the network structure is unique. 
Since the network structure can be determined by angle 
constraints and displacement constraints, to analyze network localizability, 
we need to obtain 
angle constraints and displacement constraints.

\subsubsection{Angle Constraint}\label{ag1}

The known anchor positions will be used to construct the angle constraints, which do not need any measurement. 
Under Assumption \ref{ass3}, $e_{ij}, e_{ik}, e_{jk}$ are known among the anchors  $i,j,k$. Then, the parameters $w_{ik}$, $w_{ki}$, $w_{ij}$, $w_{ji}$, $w_{jk}$, $w_{kj}$ in the angle constraints \eqref{ang1}-\eqref{aa2} can be obtained by \eqref{para}. 
\subsubsection{Displacement Constraint} Different from the angle constraints that are constructed without any measurement, 
the local relative measurements will be used to construct the displacement constraints \eqref{root}.
Under Assumption \ref{ass3}, 
each free node $i$ can communicate with its four neighbors $j,k,h,l$. Then,
the displacement constraint among $i,j,k,h,l$ can be obtained by the mixed types of measurements. Without loss of generality, 
we consider two cases: $i \in \mathcal{R}$ and $i \in \mathcal{B} \cup\mathcal{A}  \cup  \mathcal{D} \cup\mathcal{R}o\mathcal{D}$. 

\textbf{Case $\mathbf{(1)}$}. Node $i \in \mathcal{R}$. Node $i$ can measure local relative positions
$e_{ij}^i, e_{ik}^i, e_{ih}^i, e_{il}^i$ shown in Fig. \ref{case}$(a)$. In this case, we do not care the measurement types of nodes $j,k,l,h$ as their measurements will not be used. 
Since $(e^i_{ij}, e^i_{ik}, e^i_{ih}, e^i_{il}) \in \mathbb{R}^{3 \times 4}$, there must be a non-zero vector $ (\mu_{ij}, \mu_{ik}, \mu_{ih}, \mu_{il})^T \in \mathbb{R}^4$ satisfying 
\begin{equation}\label{wmi}
\left[ \!
\begin{array}{c c c c}
e_{ij}^{i} & e_{ik}^{i}  & e_{ih}^{i} &  e_{il}^{i} \\
\end{array}
\right]  \left[ \!
	\begin{array}{c}
	\mu_{ij} \\
	\mu_{ik} \\
	\mu_{ih} \\
	\mu_{il}
	\end{array}
	\right] = \mathbf{0}.
\end{equation}

Note that
$e_{ij} = Q_ie_{ij}^{i}, e_{ik} = Q_ie_{ik}^{i},e_{ih} = Q_ie_{ih}^{i}, e_{il} = Q_ie_{il}^{i}$. Left multiplying $Q_i$ on both sides of \eqref{wmi}, we can obtain the displacement constraint 
\begin{equation}\label{r1}
    \mu_{ij}e_{ij}+\mu_{ik}e_{ik}+\mu_{ih}e_{ih} + \mu_{il}e_{il}= \mathbf{0}.
\end{equation}

From  \eqref{r1}, we can know that
the parameters $\mu_{ij}, \mu_{ik}$, $\mu_{ih}, \mu_{il}$ in
\eqref{r1} can be calculated by solving \eqref{wmi}.

\textbf{Case $\mathbf{(2)}$}. Node $i \in \mathcal{B} \cup\mathcal{A}  \cup  \mathcal{D} \cup\mathcal{R}o\mathcal{D}$. There are two situations: $(\mathbf{2.1})$ at least one neighbor of node $i$ in $\mathcal{R}$, $(\mathbf{2.2})$ no neighbor of node $i$ in $\mathcal{R}$.

\textbf{Case} $(\mathbf{2.1})$. Consider at least one neighbor of node $i$ in $\mathcal{R}$, say $j \in \mathcal{R}$ shown in Fig. \ref{case}$(b)$. Node $i$ can obtain $e_{ji}^j, e_{jk}^j, e_{jh}^j, e_{jl}^j$ from node $j$ by communication. 
Similar to \eqref{wmi} and \eqref{r1},
we can obtain a displacement constraint shown as
\begin{equation}\label{r2}
    \mu_{ji}e_{ji}+\mu_{jk}e_{jk}+\mu_{jh}e_{jh} + \mu_{jl}e_{jl}= \mathbf{0},
\end{equation}
where the non-zero vector $ (\mu_{ji}, \mu_{jk}, \mu_{jh}, \mu_{jl})^T \in \mathbb{R}^4$ satisfying \eqref{r2} can be calculated by solving its equivalent equation \eqref{wmj}.
\begin{equation}\label{wmj}
\left[ \!
\begin{array}{c c c c}
e_{ji}^{j} & e_{jk}^{j}  & e_{jh}^{j} &  e_{jl}^{j} \\
\end{array}
\right]  \left[ \!
	\begin{array}{c}
	\mu_{ji} \\
	\mu_{jk} \\
	\mu_{jh} \\
	\mu_{jl}
	\end{array}
	\right] = \mathbf{0}.
\end{equation}

\begin{remark}
For the case $(\mathbf{2.1})$, node $i$ makes use of node $j$'s measurements to construct the displacement constraint \eqref{r2}. 
\end{remark}

\textbf{Case} $(\mathbf{2.2})$. Consider no neighbor of node $i$ in $\mathcal{R}$, that is, $i, j, k, h, l \in  \mathcal{B} \cup\mathcal{A} \cup \mathcal{D}  \cup \mathcal{R}o\mathcal{D}$. 
Under Assumption \ref{ass3}, node $i$ and its four neighboring nodes $j,k,h,l$ are non-colinear. Hence, there must be six triangles among the nodes $i,j,k,h,l$ shown in Fig. \ref{case}$(c)$, where  
the six triangles are: $\bigtriangleup_{ijl}(p)$, $\bigtriangleup_{ikl}(p)$, $\bigtriangleup_{ihl}(p)$, $\bigtriangleup_{jkl}(p)$, $\bigtriangleup_{jhl}(p)$, $\bigtriangleup_{khl}(p)$. 
It is shown in Section \ref{disdis} that the parameters $\mu_{ij}, \mu_{ik}, \mu_{ih}, \mu_{il}$ in the 
displacement constraint \eqref{r1} can be obtained by the distance matrix $D$ \eqref{distance} through \text{Algorithm \ref{array-sum}}, but the  distance matrix $D$ \eqref{distance} is usually unavailable in a network with mixed local relative measurements.
We find that
the following ratio-of-distance matrix $D_r$ \eqref{ratiom} can be obtained by using the mixed local relative measurements. 
\begin{equation}\label{ratiom}
D_r = \frac{D}{d_{ij}^2}.   
\end{equation}

From \textit{Lemma \ref{inva}}, we can know that
the scaling motions can preserve the invariance of displacement constraint \eqref{r1}, i.e.,  a network with ratio-of-distance matrix $D_r$ \eqref{ratiom} has the same displacement constraint as the network with distance matrix $D$ \eqref{distance}. Hence, parameters $\mu_{ij}, \mu_{ik}, \mu_{ih}, \mu_{il}$ in  \eqref{r1} can be 
obtained by the ratio-of-distance matrix $D_r$ \eqref{ratiom} through \text{Algorithm \ref{array-sum}}.  Next, we will introduce how to obtain the ratio-of-distance matrix $D_r$ \eqref{ratiom} based on the six triangles $\bigtriangleup_{ijl}(p)$, $\bigtriangleup_{ikl}(p)$, $\bigtriangleup_{ihl}(p)$, $\bigtriangleup_{jkl}(p)$, $\bigtriangleup_{jhl}(p)$, $\bigtriangleup_{khl}(p)$ shown in Fig. \ref{case}$(c)$.

We will use triangle $\bigtriangleup_{ijl}(p)$ as an example to show how to obtain the ratio-of-distances $\frac{d_{il}}{d_{ij}}, \frac{d_{jl}}{d_{ij}}$  in $D_r$. In the {case} $(\mathbf{2.2})$,
for the triangle $\bigtriangleup_{ijl}(p)$ with nodes $p_i, p_j, p_l$, there are two categories: (a) 
the 
nodes $p_i, p_j, p_l$ have different types of measurements, (b) 
at least two of the nodes  $p_i, p_j, p_l$ have the same type of measurement.

\begin{enumerate}[(a)]

\item The nodes $p_i, p_j, p_l$ have different types of measurements. Since $i,j,l \in \mathcal{B} \cup\mathcal{A} \cup \mathcal{D}  \cup\mathcal{R}o\mathcal{D}$, 
the following four cases: (a1)
$i,j,l \in \mathcal{B} \cup \mathcal{A} \cup\mathcal{D}$, (a2) 
$i,j,l \in \mathcal{B} \cup \mathcal{A} \cup \mathcal{R}o\mathcal{D}$, (a3) $i,j,l \in \mathcal{B} \cup \mathcal{D} \cup \mathcal{R}o\mathcal{D}$, (a4) $i,j,l \in \mathcal{A} \cup \mathcal{D} \cup \mathcal{R}o\mathcal{D}$, cover all possibilities. Without loss of generality, for the cases (a1) and (a2), suppose node $i$ can measure local relative bearings $g_{ij}^i, g_{il}^i$
and node $j$ can measure angle $\theta_{ijl}$. Then, we have $\theta_{jil}= {g_{ij}^i}^Tg_{il}^i$  and  $\theta_{ilj}= \pi - \theta_{jil} -  \theta_{ijl}$. The ratio-of-distances can be obtained by $\frac{d_{il}}{d_{ij}}=\frac{\sin \theta_{ijl}}{\sin \theta_{ilj}}$ and $\frac{d_{jl}}{d_{ij}}=\frac{\sin \theta_{jil}}{\sin \theta_{ilj}}$. For the cases (a3) and (a4), suppose node $i$ can measure distances $d_{ij}, d_{il}$, and node $j$ can measure the ratio of distance $\frac{d_{jl}}{d_{ij}}$. The ratio-of-distances $\frac{d_{il}}{d_{ij}}, \frac{d_{jl}}{d_{ij}}$ can also be obtained. 
    
\item At least two of the nodes  $p_i, p_j, p_l$ have the same type of measurement. Since $i,j,l \in \mathcal{B} \cup\mathcal{A} \cup \mathcal{D}  \cup\mathcal{R}o\mathcal{D}$,  
the following four cases: (b1) at least two nodes can measure local relative bearing, (b2) at least two nodes can measure angle, (b3) at least two nodes can measure distance, (b4) at least two nodes can measure ratio-of-distance, cover all possibilities. Without loss of generality, for the case (b1), suppose the nodes $i,j$ can measure local relative bearings $g_{ij}^i, g_{il}^i$ and $g_{ji}^j, g_{jl}^j$. Then, we have $\theta_{jil}= {g_{ij}^i}^Tg_{il}^i$, $\theta_{ijl}= {g_{ji}^j}^Tg_{jl}^j$ and $\theta_{ilj}= \pi -  \theta_{jil}- \theta_{ijl}$. The ratio-of-distances can be obtained by $\frac{d_{il}}{d_{ij}}=\frac{\sin \theta_{ijl}}{\sin \theta_{ilj}}$ and $\frac{d_{jl}}{d_{ij}}=\frac{\sin \theta_{jil}}{\sin \theta_{ilj}}$. For the case (b2), suppose the nodes $i,j$ can measure angles $\theta_{jil}$ and $\theta_{ijl}$. Then, we have
$\theta_{ilj}= \pi -  \theta_{jil}- \theta_{ijl}$. The ratio-of-distances can also be obtained by $\frac{d_{il}}{d_{ij}}=\frac{\sin \theta_{ijl}}{\sin \theta_{ilj}}$ and $\frac{d_{jl}}{d_{ij}}=\frac{\sin \theta_{jil}}{\sin \theta_{ilj}}$.  
For the cases (b3) and (b4) where at least two nodes can measure distance or ratio-of-distance, it is clear that the ratio-of-distances $\frac{d_{il}}{d_{ij}}, \frac{d_{jl}}{d_{ij}}$ can be obtained.
\end{enumerate}

Hence,  the ratio-of-distances $\frac{d_{il}}{d_{ij}}, \frac{d_{jl}}{d_{ij}}$ in $D_r$ can be obtained by the triangle $\bigtriangleup_{ijl}(p)$.
The rest ratio-of-distances in $D_r$ can also be obtained by the triangles $\bigtriangleup_{ikl}(p)$, $\bigtriangleup_{ihl}(p)$, $\bigtriangleup_{jkl}(p)$, $\bigtriangleup_{jhl}(p)$, $\bigtriangleup_{khl}(p)$, i.e., the ratio-of-distance matrix $D_r$ is available. Then, the displacement constraint \eqref{r1} can be
obtained by the ratio-of-distance matrix $D_r$ \eqref{ratiom} through \text{Algorithm \ref{array-sum}}. The implementation of constructing the displacement constraint is given in Fig. \ref{chat}.

\begin{figure}[t]
\centering
\includegraphics[width=1\linewidth]{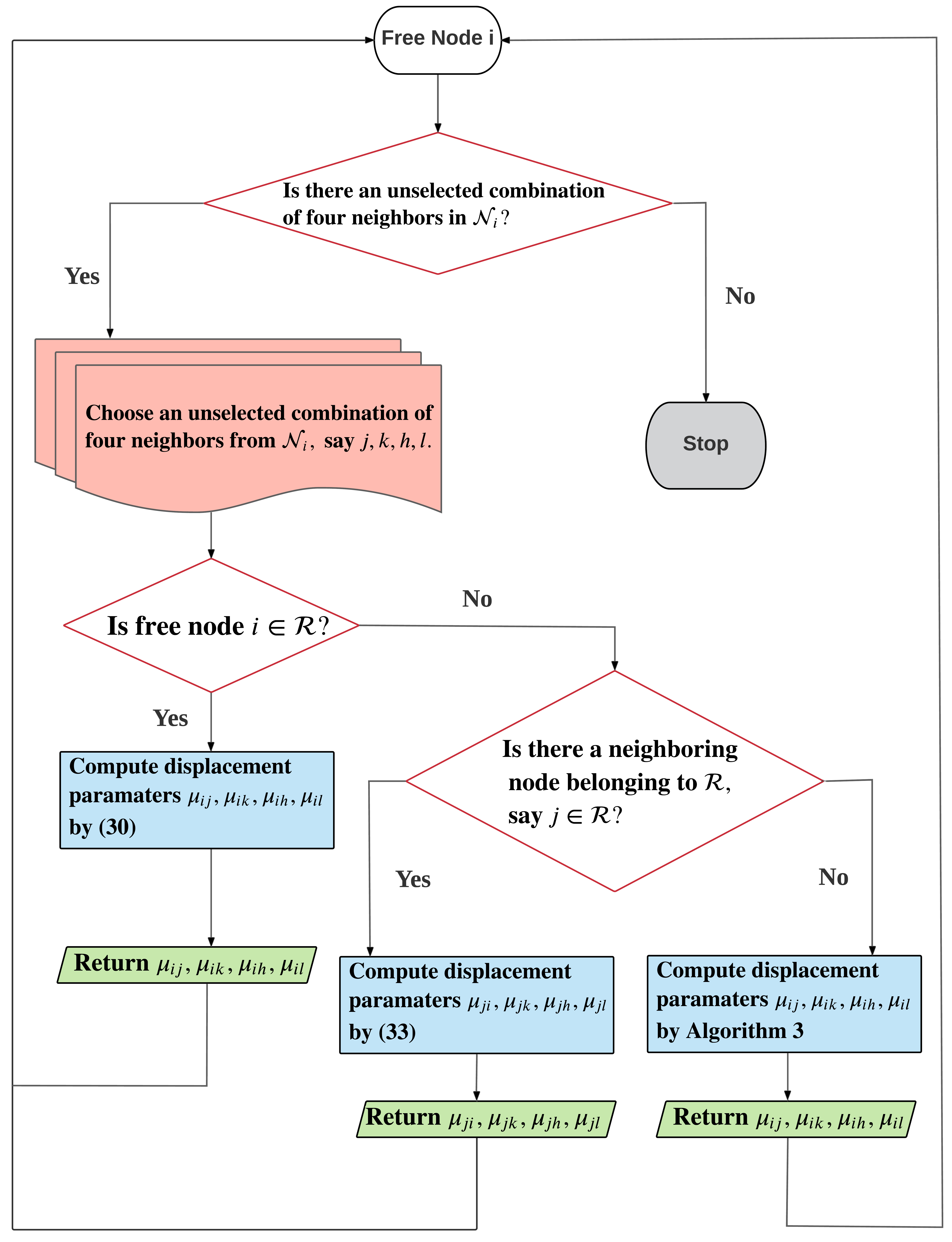}
\caption{Construction of the displacement constraint. }
\label{chat}
\end{figure}

After obtaining the angle constraints and displacement constraints, we are now ready to formulate the network localization problem and explore the network localizability.

\subsection{Network Localizability}

The angle constraints are used for describing the distribution of the anchors, and the displacement constraints are used to estimate the free nodes by the known anchor nodes. The network localization problem is formulated as
\begin{equation}\label{cost}
\begin{array}{ll}
& \hspace*{-0.5cm}    \min\limits_{\hat p \in \mathbb{R}^{3n}} J(\hat p) \!=\! \|L_{\mathcal{X}_{\mathcal{G}}}(\hat p)\|^2, \\
& \hspace*{-0.5cm} \text{subject} \  \text{to} \ \ \hat{p}_i =p_i, \ \ i \in \mathcal{V}_a,
\end{array}
\end{equation}
where $\hat p = (\hat p_1^T, \cdots, \hat p_n^T)^T$ are the estimates of $p= ( p_1^T, \cdots,  p_n^T)^T$, and $L_{\mathcal{X}_{\mathcal{G}}}(\cdot)$ is  displacement function \eqref{func}.  If there is a unique solution to \eqref{cost}, the 
network ($\mathcal{G}, p$) is called localizable.

\begin{remark}
$L_{\mathcal{X}_{\mathcal{G}}}(\cdot)$ is the displacement function consisting of available displacement constraints.  Since the true positions $p$ of all the nodes satisfy $L_{\mathcal{X}_{\mathcal{G}}}( p)=0$, we have $\| L_{\mathcal{X}_{\mathcal{G}}}(\hat p)\|_2 = 0 $ if $\hat p = p$. Hence, the network localization problem is equivalent to minimizing the cost function $\min\limits_{\hat p \in \mathbb{R}^{3n}} J(\hat p) \!=\! \|L_{\mathcal{X}_{\mathcal{G}}}(\hat p)\|^2$ given the known anchor positions.
\end{remark}

From \eqref{anre}, we obtain
\begin{equation}\label{anre1}
4\| B_{\Upsilon_{\mathcal{G}}}(\hat p) \|^2 \!+\! J(\hat p)= \hat p^TR(\hat p)^TR(\hat p)\hat p.  
\end{equation}

For the angle function $B_{\Upsilon_{\mathcal{G}}}(\cdot)$ \eqref{funa} constructed among the known anchors,
it yields $B_{\Upsilon_{\mathcal{G}}}(p) = B_{\Upsilon_{\mathcal{G}}}([\begin{array}{c}
    p_{a}  \\
    \mathbf{0}  
    \end{array}])=B_{\Upsilon_{\mathcal{G}}}(\hat p)=0$, i.e., $\|B_{\Upsilon_{\mathcal{G}}}(\hat p) \|^2=0$. Then, \eqref{anre1} becomes
    \begin{equation}\label{rt1}
     J(\hat p)= \hat p^TR(\hat p)^TR(\hat p)\hat p. 
    \end{equation}

From \eqref{infid} and \eqref{infi},  we can know that the matrix $R(p)$ are determined by the anchor positions $p_a$ and parameters (such as $w_{ik}, w_{ki}$ in \eqref{funa} and $\mu_{ij}, \mu_{uk}, \mu_{ih}, \mu_{il}$ in \eqref{func}) in the angle constraints and displacement constraints. Hence, 
we have $R(p) = R([\begin{array}{c}
    p_{a}  \\
    \mathbf{0}  
    \end{array}]) = R(\hat p)$. Then, \eqref{rt1} becomes
 \begin{equation}\label{rt2}
     J(\hat p)= \hat p^TR( p)^TR(p)\hat p. 
\end{equation}    

Define
information matrix $M \in \mathbb{R}^{3{n} \times 3{n}}$ of a network with mixed local relative measurements as
\begin{equation}\label{info}
M = R(p)^T R(p),
\end{equation}
which can be partitioned as 
\begin{equation}\label{block}
    M = \left[\begin{array}{cc}
    M_{aa} & M_{af} \\
    M_{fa}  &  M_{ff}  
    \end{array}\right],
\end{equation}
where $M_{aa} \in \mathbb{R}^{3{n_a} \times 3{n_a}}$, $M_{fa}^T = M_{af} \in \mathbb{R}^{3{n_a} \times 3{n_f}}$, and $M_{ff} \in \mathbb{R}^{3{n_f} \times 3{n_f}}$. Note that $\hat p =(p_a^T, \hat p_f^T)^T$. 
Based on \eqref{rt2} and \eqref{block}, \eqref{cost} can be rewritten as 
\begin{equation}\label{mini}
\min\limits_{\hat p_f \in \mathbb{R}^{3n_f}} \tilde J(\hat p_f) =  \hat{p}_f^T M_{ff} \hat{p}_f+  2{p}_a^T M_{af} \hat{p}_f+ {p}_a^T M_{aa} {p}_a. 
\end{equation}

Note that any minimizer $\hat p_f^*$ in \eqref{mini} satisfies
\begin{equation}\label{mini1}
\bigtriangledown_{{\hat p_f}^*} \tilde J({\hat p_f}^*)=M_{ff} {{\hat p_f}^*} + M_{fa}p_a= \mathbf{0}.  \end{equation}

From \eqref{mini1}, we can obtain the algebraic and graph conditions for network localizability shown in the following Theorem \ref{the3} and Theorem \ref{t3}, respectively.

\begin{theorem}\label{the3}
{Suppose Assumption \ref{ass3} holds.} A  network ($\mathcal{G}, p$) with mixed types of measurements in $\mathbb{R}^3$ is localizable if and only if the matrix $M_{ff}$ is nonsingular. 
\end{theorem}
\begin{proof}
(Sufficiency)  From \eqref{mini1}, we can know that the matrix $M_{ff}$ must be nonsingular. (Necessity) If the matrix $M_{ff}$ is nonsingular, we have $\hat p_f^*=- M_{ff}^{-1}M_{fa}{p}_a$. Since $p \in \text{Null}(R(p))$ and $M = R(p)^T R(p)$, we have $p \in \text{Null}(M)$, that is, $M_{ff} p_f + M_{fa}p_a=\mathbf{0}$. Since $M_{ff}$ is nonsingular, we have $p_f=- M_{ff}^{-1}M_{fa}{p}_a$. Hence, $\hat p_f^*$ equals the true positions $p_f$. 
\end{proof}

\begin{figure}[t]
\centering
\includegraphics[width=1\linewidth]{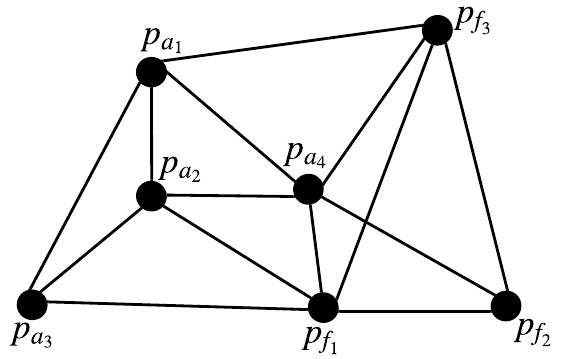}
\caption{3-D localizable network. The anchor nodes are $p_{a_1},p_{a_2},p_{a_3},p_{a_4}$ and the free nodes are $p_{f_1},p_{f_2},p_{f_3}$. The black nodes mean that they can measure only one of the five types of local relative measurements.}
\label{t1}
\end{figure}

From Theorem \ref{the3}, we have $p_f=- M_{ff}^{-1}M_{fa}{p}_a$
if a network ($\mathcal{G}, p$) with mixed types of measurements is localizable. A simple 3-D localizable network consisting of four anchor nodes $p_{a_1}=(0,0,20)^T$, $p_{a_2}=(0,0,0)^T$, $p_{a_3}=(10,-10,0)^T$, $p_{a_4}=(0,20,0)^T$ and three free nodes $p_{f_1}=(10,20,0)^T$, $p_{f_2}=(10,40,0)^T$, $p_{f_3}=(\frac{5}{2},30,30)^T$ is given in Fig. \ref{t1}. 
Based on the displacement constraints $e_{{a_2}{f_1}}-\frac{2}{3}e_{{a_3}{f_1}}-e_{{a_4}{f_1}}=\mathbf{0}$, $e_{{a_3}{f_2}}-\frac{5}{2}e_{{f_1}{f_2}}=\mathbf{0}$, and $-\frac{3}{4}e_{{a_1}{f_3}}+\frac{3}{8}e_{{a_4}{f_3}}+\frac{7}{8}e_{{f_1}{f_3}}-e_{{f_2}{f_3}}=\mathbf{0}$, we can obtain the corresponding matrices $M_{ff}$ and $M_{fa}$ shown as
\begin{equation}\label{nume}
\begin{array}{ll}
     & M_{ff}\!=\! \left[ \!
	\begin{array}{lll}
	\frac{1413}{311} & -\frac{61}{24} & \frac{7}{16} \\
	-\frac{61}{24} & 2 & -\frac{1}{2} \\
	\frac{7}{16} & -\frac{1}{2} & \frac{1}{4} 
	\end{array}
	\right] \otimes I_3, \\
	\\
    &  M_{fa}\!=\! \left[ \!
	\begin{array}{llll}
	-\frac{21}{32} & \frac{3}{2} & -\frac{19}{9} & -\frac{75}{64} \\
	\frac{3}{4} & 0 & \frac{2}{3} & -\frac{3}{8} \\
	-\frac{3}{8} & 0 & 0 & \frac{3}{16}
	\end{array}
	\right] \otimes I_4.
\end{array}
\end{equation}

It can be verified that $p_f= -M_{ff}^{-1}M_{fa}{p}_a$, where $p_a=(p_{a_1}^T,p_{a_2}^T,p_{a_3}^T,p_{a_4}^T)^T$ and $p_f=(p_{f_1}^T,p_{f_2}^T,p_{f_3}^T)^T$. Hence, the network in Fig. \ref{t1} is localizable.

\begin{theorem}\label{t3}
{Suppose Assumption \ref{ass3} holds.} 
A 3-D network ($\mathcal{G}, p$) with mixed types of measurements is localizable if the following conditions hold:
\begin{enumerate}
    \item ($\mathcal{G}, p$) is infinitesimally angle-displacement rigid;
    \item There are at least three non-colinear anchor nodes.
\end{enumerate}
\end{theorem}
\begin{proof}
From \textit{Lemma \ref{inva}}, $\text{Null}(R(p)) = \text{Span}\{ \mathbf{1}_n\otimes I_3, p, (I_n \otimes A)p, A+A^T =\mathbf{0},  A \in \mathbb{R}^{3 \times 3} \}$
if ($\mathcal{G}, p$) is
infinitesimally angle-displacement rigid. Thus,
the nonzero angle-displacement infinitesimally motion
is 
\begin{equation}
\delta p= sp + 
r(I_n \otimes  {A})p +   \mathbf{1}_n\otimes \mathbf{t},   
\end{equation}
where $s, r \in \mathbb{R}, \mathbf{t} \in \mathbb{R}^3, s^2+\|rA\|+ \|\mathbf{t}\| \neq 0$.  For the anchor nodes $p_a$, we have $\delta p_a= sp_a+ 
r(I_{n_a} \otimes  {A})p_a +   \mathbf{1}_{n_a}\otimes \mathbf{t}$. 

\begin{enumerate}[(a)]
\item First, we will prove $\delta p_a \neq \mathbf{0}$ by contradiction. 
If $\delta p_a = \mathbf{0}$, we have 
\begin{equation}\label{s1}
  (sI_3+rA)p_i=  -\mathbf{t}, \ \ i \in \mathcal{V}_a. 
\end{equation}

Since $A+A^T =\mathbf{0}$, matrix $A$ is a real skew-symmetric matrix and each eigenvalue of $A$ is either zero or pure imaginary. 
If $s \neq 0$, we have $\text{Rank}(sI_3+rA)=3$. From \eqref{s1}, we have $p_i= -(sI_3+rA)^{-1}\mathbf{t}, i \in \mathcal{V}_a$, i.e., the anchor nodes collocate, which contradicts Assumption \ref{ass3} that no two anchor nodes are collocated. If $s = 0$, we have $rA, \mathbf{t} \neq \mathbf{0}$, i.e., $\text{Rank}(rA) \ge 2$. From \eqref{s1}, we can know that the anchor nodes 
are colinear or collocated, which contradicts 
that no two anchor nodes are collocated and there are at least three non-colinear anchors. Hence, $\delta p_a \neq \mathbf{0}$.

\item Second, we will prove that the matrix $M_{ff}$ is nonsingular if $\delta p_a \neq 0$. We only need to show that $\delta p_a = 0$ if the matrix $M_{ff}$ is singular. 
If the matrix $M_{ff}$ is singular, there must be a nonzero $\delta p_f$ such that $M_{ff}\delta p_f=\mathbf{0}$. Let $\delta p = (\delta p_a^T, \delta p_f^T)^T$ with $\delta p_a=\mathbf{0}$. We have $\delta p^TM\delta p=\delta p_f^TM_{ff}\delta p_f=0$. 
\end{enumerate}

Based on (a) and (b), we can know that the matrix $M_{ff}$ must be nonsingular if  ($\mathcal{G}, p$) is infinitesimally angle-displacement rigid. Hence,  ($\mathcal{G}, p$) is localizable from Theorem \ref{the3}.
\end{proof}

The graph condition in Theorem \ref{t3} is sufficient but not necessary.  A 3-D localizable network may not be infinitesimally angle-displacement rigid. As shown in Fig. \ref{t1}, the anchor nodes $p_{a_1}, p_{a_2}, p_{a_3}, p_{a_4}$ are enough to localize the free nodes $p_{f_1}, p_{f_2}, p_{f_3}$. If we add additional anchor node $ p_{a_5}$ with edges $(p_{a_1}, p_{a_5})$ and $(p_{a_2}, p_{a_5})$ to the network, the 
network will not be infinitesimally angle-displacement rigid but is still localizable.
Next, we will answer the question of how to localize the free node in a distributed way.

\subsection{Distributed Localization Algorithm}

To achieve distributed network localization,
each free node $i \in \mathcal{V}_f$ implements the following
distributed localization algorithm, i.e.,
\begin{equation}\label{law}
\dot  {\hat{p}}_i  = \sum\limits_{(i,j,k,h,l) \in  \mathcal{X}_{\mathcal{G}}}N(i,j,k,h,l)+ \sum\limits_{(j,i,k,h,l) \in  \mathcal{X}_{\mathcal{G}}} \bar N(j,i,k,h,l), 
\end{equation}
where
\begin{equation}\label{n1}
\begin{array}{ll}
     & N(i,j,k,h,l) = (\mu_{ij}+\mu_{ik}+\mu_{ih}+\mu_{il}) \cdot  [\mu_{ij}(\hat{p}_j - \hat{p}_i)\!+\! \\
     &\mu_{ik}(\hat{p}_k - \hat{p}_i)\!+\!\mu_{ih}(\hat{p}_h - \hat{p}_i)\!+\!\mu_{il}(\hat{p}_l - \hat{p}_i)],
\end{array}
\end{equation}
and
\begin{equation}\label{n2}
    \begin{array}{ll}
     & \bar N(j,i,k,h,l) = \mu_{ji}^2(\hat{p}_j - \hat{p}_i)\!+\! \mu_{ji}\mu_{jk}(\hat{p}_j - \hat{p}_k)\!+\! \\
     & \mu_{ji}\mu_{jh}(\hat{p}_j - \hat{p}_h)\!+\!\mu_{ji}\mu_{jl}(\hat{p}_j - \hat{p}_l).
\end{array}
\end{equation}

$N(i,j,k,h,l)$ corresponds to the displacement constraint \eqref{r1} and $\bar N(j,i,k,h,l)$ corresponds to the displacement constraint \eqref{r2}. Based on \eqref{law}, we have 
\begin{equation}\label{dis}
\dot  {\hat{p}}_f = -M_{ff}{\hat{p}}_f - M_{fa}p_a.
\end{equation}

\begin{theorem}\label{ts3}
Suppose Assumption \ref{ass3} holds.
Under
the distributed localization algorithm 
\eqref{law} with arbitrary generated initial estimates, the position estimation errors of 
the free nodes asymptotically 
converge to zero if the following conditions hold:
\begin{enumerate}
    \item ($\mathcal{G}, p$) is infinitesimally angle-displacement rigid;
    \item There are at least three non-colinear anchor nodes.
\end{enumerate}
\end{theorem}
\begin{proof}
From \eqref{info}, 
we can conclude that the matrix $M_{ff}$ is positive semi-definite.  Since ($\mathcal{G}, p$) is 
infinitesimally angle-displacement rigid, the matrix $M_{ff}$ must be positive definite from Theorem \ref{t3}.
Consider the Lyapunov function $V = \frac{1}{2}\| {\hat{p}}_f - {{p}}_f \|^2$, it yields
\begin{equation}
\begin{array}{ll}
\dot V &= (\hat p_f - p_f)^T \dot  {\hat{p}}_f \\
& = (\hat p_f - p_f)^T(-M_{ff}{\hat{p}}_f - M_{fa}p_a) \\
& = (\hat p_f - p_f)^T(-M_{ff}{\hat{p}}_f +M_{ff}p_f) \\
& = -(\hat p_f - p_f)^TM_{ff}(\hat p_f - p_f) <0, \ \text{if} \ \hat p_f \neq p_f.
\end{array}
\end{equation}
Hence, given arbitrary initial estimates $\hat p_i(0), i \in \mathcal{V}_f$,  $\hat p_i(t)$ converges to $p_i$ asymptotically. 
\end{proof}

\begin{algorithm}
\caption{The sequential distributed localization algorithm}
\label{sequ}
\begin{algorithmic}[1]
\State  \textbf{Initialization:}
Each node starts sending and receiving local relative  measurements to or from its neighboring nodes. Each node has two modes: localized and unlocalized. The anchor nodes are set as localized mode, and the free nodes are set as unlocalized mode with initial estimates. Denote $\mathcal{N}_{il}$  as the set of localized neighbors of node $i$, and the $4$-combinations from the set $\mathcal{N}_{il}$ is denoted by $C_i$;
\State For any node $i$: 
\State  \ \  \ \   \textbf{While} node $i$ is in localized mode, \textbf{do}
\State  \ \  \ \ \ \ Sending its position and local relative  measurements 
\State  \ \  \ \ \ \ to its neighboring nodes;
\State  \ \  \ \    
\textbf{End}
\State  \ \  \ \    \textbf{While}  node $i$ is in unlocalized mode, \textbf{do}
\State  \ \  \ \ \ \ Updating the set $\mathcal{N}_{il}$;
\State  \ \  \ \ \ \   \textbf{If} $|\mathcal{N}_{il}| \ge 4$, for each $4$-combination of $C_i$ such as 
\State  \ \  \ \ \ \ localized neighbors $j,k,h,l$, \textbf{do}
\State  \ \  \ \ \ \  \ \  Based on the mixed noisy measurements and the 
\State  \ \  \ \ \ \  \ \  method shown in Section \ref{conc}, 
constructing the
\State  \ \  \ \ \ \  \ \   displacement constraint among  $i,j,k,h,l$;
\State  \ \  \ \ \ \  \ \   Obtaining displacement parameters $\mu_{ij}, \mu_{ik}, \mu_{ih}, $
\State  \ \  \ \ \ \  \ \    $\mu_{il}$, or $\mu_{ji}, \mu_{jk}, \mu_{jh}, \mu_{jl}$;
\State  \ \  \ \ \ \  \ \   \textbf{If} $\mu_{ij} \!+\! \mu_{ik}\!+\!\mu_{ih}\!+\! \mu_{il} \neq 0$ or $\mu_{ji} \neq 0$, \textbf{then}
\State  \ \  \ \ \ \  \ \ \ \  $\dot  {\hat{p}}_i \!=\! N(i,j,k,h,l)$ \eqref{n1} or
$ \bar N(j,i,k,h,l)$ \eqref{n2}; 
\State  \ \  \ \ \ \  \ \  \ \  Obtaining the estimate ${\hat p}^*_i$ of node $i$; 
\State  \ \  \ \ \ \  \ \  \ \  Switching to localized mode;
\State  \ \  \ \ \ \  \ \  \ \   \textbf{Break}
\State  \ \  \ \ \ \  \ \   \textbf{End}
\State  \ \  \ \ \ \  \textbf{End}
\State  \ \  \ \ \textbf{End}
\end{algorithmic}
\end{algorithm}

If there is measurement noise,
the measurement noise will influence the parameters in the displacement constraints such as $\mu_{ij}, \mu_{ik}, \mu_{ih}, \mu_{il}$ in \eqref{func}. Note that the matrices $M_{ff}$ and $M_{fa}$ in \eqref{dis} are only determined by the parameters in the displacement constraints. Hence, the measurement noise will influence matrices $M_{ff}$ and $M_{fa}$. The noisy $\bar M_{ff}$ and $\bar M_{fa}$ are
\begin{equation}\label{noi}
    \bar M_{ff} = M_{ff} + \Delta M_{ff}, \ \ \bar M_{fa} = M_{fa} + \Delta M_{fa},
\end{equation}
where $\Delta M_{ff}, \Delta M_{fa}$ are error matrices. Then, we have
\begin{equation}\label{e1}
    \bar M_{ff} \hat p^*_f + \bar M_{fa} p_a = \mathbf{0},
\end{equation}
where $\hat p^*_f$ is the positions of the free nodes satisfying the noisy measurements and known anchor positions.  If the noisy matrix $\bar M_{ff}$ is nonsingular,  the solution $\hat p_f^*$ to \eqref{e1} is unique.

\begin{lemma}
A  network ($\mathcal{G}, p$) with mixed types of noisy measurements in $\mathbb{R}^3$ is localizable if the  following conditions hold:
\begin{enumerate}
    \item ($\mathcal{G}, p$) is infinitesimally angle-displacement rigid;
    \item The error matrix $\Delta  M_{ff}$  satisfies
\begin{equation}\label{perror}
\| \Delta  M_{ff} \| <  \lambda_{\min}(M_{ff}).
\end{equation}
\end{enumerate} 
\end{lemma}

\begin{proof}
From Theorem \ref{the3}, $M_{ff}$ is nonsingular if ($\mathcal{G}, p$) is infinitesimally angle-displacement rigid. Then, we have
$\bar{M}_{ff} = M_{ff} (I \!+\! M_{ff}^{-1}\Delta  M_{ff})$. 
Since
$\| M_{ff}^{-1} \Delta  M_{ff}   \| \le \| M_{ff}^{-1} \| \|  \Delta M_{ff} \|= \frac{\| \Delta M_{ff} \|}{\lambda_{\min}(M_{ff})}$,
the matrix $\bar M_{ff}$ must be nonsingular if the condition in \eqref{perror} holds.  From \eqref{info}, 
we can conclude that $\bar M_{ff}$ must also be positive definite. Then, under the proposed distributed localization algorithm $\dot  {\hat{p}}_f = - \bar M_{ff}{\hat{p}}_f - \bar M_{fa}p_a$, the estimate ${\hat{p}}_f$ will converge to ${\hat{p}^*}_f$ shown as
\begin{equation}\label{uu1}
 \hat p_f^*= - \bar M_{ff}^{-1}\bar M_{fa}p_a.
\end{equation}
\end{proof}

\begin{remark}
The condition in \eqref{perror} is only used to  guarantee the noisy matrix $\bar M_{ff}$ to be nonsingular, which is sufficient but not necessary. Hence, the noisy matrix $\bar M_{ff}$ may still be nonsingular even if the condition in (53) does not hold, that is, the network may still be localizable under the proposed scheme even if the condition in (53) does not hold.
\end{remark}

From \textit{Lemma 2}, we can know that the displacement constraints are invariant to the translation of the entire network. If  
we change the original point of the global coordinate frame or translate the whole network, 
the matrices $\bar{M}_{fa}, \bar{M}_{ff}, M_{ff},$ $M_{fa}$ will remain unchanged, i.e.,  
\begin{align}\label{c11}
& M_{ff}(p_f- \mathbf{1}_{n_f} \otimes  p_o)+ M_{fa}(p_a-  \mathbf{1}_{n_a} \otimes p_o)= \mathbf{0}, \\
\label{c1}
& \bar M_{ff} (\hat p^*_f - \mathbf{1}_{n_f} \otimes  p_o) + \bar M_{fa}(p_a  - \mathbf{1}_{n_a} \otimes  p_o)= \mathbf{0}, 
\end{align}
where $p_o \in \mathbb{R}^3$ is the original point of any local coordinate frame. If the condition in \eqref{perror} holds, based on \eqref{uu1} and \eqref{c1}, we have 
\begin{equation}\label{c3}
\begin{array}{ll}
      & \|\hat{p}_f^* - p_f \| \\
      &=\|  - \bar{M}_{ff}^{-1}\bar{M}_{fa}p_a -p_f\| \\
     & = \|  - \bar{M}_{ff}^{-1}\bar{M}_{fa}(p_a-  \mathbf{1}_{n_a} \otimes p_o) -(p_f - \mathbf{1}_{n_f} \otimes  p_o)\|. 
\end{array}
\end{equation}

From \eqref{c3}, we can know that
the position estimation error remains unchanged if we change the original point of the global coordinate frame or translate the whole network. Actually, the position estimation error is only determined by the noisy measurements and the structure of the network. Then, we can obtain the error bound of position estimates of the free nodes shown as
\begin{equation}\label{c2}
\begin{array}{ll}
       & \|\hat{p}_f^* - p_f\| \\
       & = \|  - \bar{M}_{ff}^{-1}\bar{M}_{fa}(p_a-  \mathbf{1}_{n_a} \otimes p_o) -( p_f - \mathbf{1}_{n_f} \otimes  p_o)\|   \\
       & \le  \frac{\|\Delta M_{fa}\|\|p_a - \mathbf{1}_{n_a} \otimes  p_o\| + \|\Delta M_{ff}\|\|p_f- \mathbf{1}_{n_f} \otimes  p_o\|}{\|I+M_{ff}^{-1}\Delta M_{ff}\|\|M_{ff}\|} \\
       & \le \frac{\|\Delta M_{fa}\| \sum\limits_{i=1}^{n_a} \|p_i -p_o \| + |\Delta M_{ff}\| \sum\limits_{i=n_a \!+\!1}^{n}
       \|p_i -p_o \|}{\|I+M_{ff}^{-1}\Delta M_{ff}\|\|M_{ff}\|}, \ p_o \in \mathbb{R}^3. 
\end{array}
\end{equation}

From \eqref{c2}, we can know that the error bound $u$ of position estimates can be obtained by minimizing the following cost function. 
\begin{equation}\label{bistan}
    u=\min\limits_{p_o \in \mathbb{R}^3} \frac{\|\Delta M_{fa}\| \sum\limits_{i=1}^{n_a} \|p_i -p_o \| + |\Delta M_{ff}\| \sum\limits_{i=n_a \!+\!1}^{n}
       \|p_i -p_o \|}{\|I+M_{ff}^{-1}\Delta M_{ff}\|\|M_{ff}\|}.
\end{equation}

The cost function \eqref{bistan} is a Fermat-Weber location problem, which can be solved by using the method in  \cite{trinh2015fermat}. Hence,
the error bound $u$ of position estimates will also
remain unchanged even if we change the original point of the global coordinate frame or translate the whole network.

The proposed distributed localization algorithm is implemented in a simultaneous way when the  noisy matrix $\bar M_{ff}$ is nonsingular.
When the  noisy matrix $\bar M_{ff}$ is not nonsingular, the free nodes can still be estimated by using a sequential distributed localization protocol shown in the above Algorithm \ref{sequ}, where there are two kinds of localization algorithms, i.e., $\dot  {\hat{p}}_i \! = \! N(i,j,k,h,l)$ \eqref{n1} and $\dot  {\hat{p}}_i \! = \! \bar N(j,i,k,h,l)$ \eqref{n2}.

If $\mu_{ij}\!+\!\mu_{ik}\!+\! \mu_{ih}\!+\! \mu_{il} \neq 0$, 
for the localization algorithm $\dot  {\hat{p}}_i \! = \! N(i,j,k,h,l)$ \eqref{n1},
consider a Lyapunov candidate 
$\bar V = \frac{1}{2}\| {\hat{p}}_i - \frac{\mu_{ij}\hat{p}^*_j+\mu_{ik}\hat{p}^*_k+\mu_{ih}\hat{p}^*_h+\mu_{il}\hat{p}^*_l}{\mu_{ij}\!+\!\mu_{ik}\!+\! \mu_{ih}\!+\! \mu_{il}}\|^2$
, it yields
\begin{equation}
\begin{array}{ll}
       {\dot {\bar V}} \hspace*{-0.3cm} & \!=\! -\! (\mu_{ij}\!+\!\mu_{ik}\!+\! \mu_{ih}\!+\! \mu_{il})^2\|{\hat{p}}_i \!-\! \frac{\mu_{ij}\hat{p}^*_j\!+\!\mu_{ik}\hat{p}^*_k\!+\!\mu_{ih}\hat{p}^*_h\!+\!\mu_{il}\hat{p}^*_l}{\mu_{ij}\!+\!\mu_{ik}\!+\! \mu_{ih}\!+\! \mu_{il}}\|^2  \\
     & \le 0.
\end{array}
\end{equation}
where $\hat{p}^*_j, \hat{p}^*_k, \hat{p}^*_h, \hat{p}^*_l$ are the  estimates of the localized neighbors $j,k,h,l$.
Hence, the estimate  of node $i$ will converge to $\hat{p}^*_i=\frac{\mu_{ij}\hat{p}^*_j+\mu_{ik}\hat{p}^*_k+\mu_{ih}\hat{p}^*_h+\mu_{il}\hat{p}^*_l}{\mu_{ij}\!+\!\mu_{ik}\!+\! \mu_{ih}\!+\! \mu_{il}}$. 
Similarly, if $\mu_{ji} \neq 0$, for the localization algorithm $\dot  {\hat{p}}_i \! = \! \bar N(j,i,k,h,l)$ \eqref{n2}, we can know that the estimate of node $i$ will converge to 
\begin{equation}
\hat{p}^*_i\!=\! \frac{\mu_{ji}\!+\!\mu_{jk}\!+\!\mu_{jh}\!+\!\mu_{jl}}{\mu_{ji}} \hat{p}^*_j\!-\! \frac{\mu_{jk}}{\mu_{ji}}\hat{p}^*_k \!-\! \frac{\mu_{jh}}{\mu_{ji}}\hat{p}^*_h\!-\! \frac{\mu_{jl}}{\mu_{ji}}\hat{p}^*_l.    
\end{equation}

In Algorithm \ref{sequ},
to implement the localization algorithm $\dot  {\hat{p}}_i=  N(i,j,k,h,l)$ or $\dot  {\hat{p}}_i= \bar N(i,j,k,h,l)$, node $i$ requires at least four localized neighbors, i.e.,  $|\mathcal{N}_{il}| \ge 4$.  Actually, when $|\mathcal{N}_{il}| = 3$ or $2$, node $i$ may also be localizable.

\begin{enumerate}[(a)]
    \item When $|\mathcal{N}_{il}| = 3$, node $i$ has three localized neighbors $j,k,h$.  The ratio-of-distance matrix among the nodes $i,j,k,h$ can be obtained by the noisy measurements. Based on the ratio-of-distance matrix, we can check whether the nodes $i,j,k,h$ are coplanar or the nodes $j,k,h$ are non-colinear by Menger determinant method shown in \eqref{cay1} and \eqref{cay2}. Note that
 a network with ratio-of-distance matrix has the same barycentric coordinate as the network with distance matrix. Hence, if
the nodes $i,j,k,h$ are coplanar and the nodes $j,k,h$ are non-colinear, based on the ratio-of-distance matrix, we can obtain the barycentric coordinate $\{ \mu_{ij}, \mu_{ik}, \mu_{ih} \}$ of node $i$ with respect to the nodes $j,k,h$ by Algorithm \ref{disa1}. Thus, the estimate $\hat p_i^*$ of node $i$ can be obtained by its localized neighbors $\hat {p}^*_{j}, \hat {p}^*_{k}, \hat {p}^*_{h}$, i.e.,
\begin{equation}
        \hat p_i^* = \mu_{ij}\hat {p}^*_{j}+\mu_{ik}\hat {p}^*_{k}+\mu_{ih}\hat {p}^*_{h}. 
\end{equation}

\item When $|\mathcal{N}_{il}| = 2$, node $i$ has two localized neighbors $j,k$. The ratio of distances $\frac{d_{ik}}{d_{ij}}$ and $\frac{d_{jk}}{d_{ij}}$ can be obtained by the noisy measurements. If $\frac{d_{ik}}{d_{ij}}+\frac{d_{jk}}{d_{ij}}=1$, or $\frac{d_{ij}}{d_{ik}}+\frac{d_{jk}}{d_{ik}}=1$, or $\frac{d_{ij}}{d_{jk}}+\frac{d_{ik}}{d_{jk}}=1$, the nodes $i,j,k$  are colinear. Without loss of generility, suppose $\frac{d_{ik}}{d_{ij}}+\frac{d_{jk}}{d_{ij}}=1$, we have
\begin{equation}\label{ds12}
    e_{ik} = \frac{d_{ik}}{d_{ij}}e_{ij}.
\end{equation}
Then, the estimate $\hat p_i^*$ of node $i$ can be obtained by its localized neighbors $\hat {p}^*_{j}, \hat {p}^*_{k}$, i.e.,
\begin{equation}
        \hat p_i^* = -\frac{d_{ik}}{d_{jk}} \hat {p}^*_{j}+\frac{d_{ij}}{d_{jk}}\hat {p}^*_{k}. 
\end{equation}
\end{enumerate}

\subsection{Discussions on the Proposed Method}

From Algorithm \ref{disa1} and \eqref{disd2}, we can know any four nodes can form a displacement constraint if they are on the same plane. Hence, 2-D network localizaion with mixed measurements only requires 
each free node has at least three neighbors rather than four neighbors required for 3-D network. Hence, when the proposed method is applied to 2-D space,
the Assumption \ref{ass3} can be relaxed as
\begin{assumption}\label{ass4}
No two nodes are collocated. Each anchor node $i$ has at least two neighboring anchor nodes $j,k$, where $i,j,k$ are mutual neighbors. Each free node $i$ has at least three neighboring nodes $j,k,h$, where $i,j,k,h$ are mutual neighbors and non-colinear.
\end{assumption}

\begin{remark}
From \eqref{wmi}, we can know that the graph can be directed if each free node can measure relative positions in its local coordinate frame. In addition, it is possible to apply the proposed method to 
a switching graph
if there exists $T$  such that for each time instant ${t_1}$, the union network $(\mathcal{G}([ t_1,  t_1+T]), p)$ is infinitesimally angle-displacement rigid, where $\mathcal{G}([ t_1,  t_1+T])= \{ \mathcal{V}, \bigcup\limits_{ t \in [ t_1,  t_1+T] } \mathcal{E}( t) \}$ is the union graph during the time interval $[ t_1,  t_1+T]$.
\end{remark}

The existing works consider distributed localization with mixed distance and bearing measurements \cite{lin2015distributed, eren2011cooperative, stacey2017role}, but \cite{lin2015distributed} is limited to 2-D space and \cite{eren2011cooperative, stacey2017role} need to know the global coordinate frame. Note that each node in \cite{lin2015distributed, eren2011cooperative, stacey2017role} can measure both distances and bearings.  The more challenging case is that 
some nodes can only measure distances or bearings. For this case, Lin develops a distributed algorithm for 2-D network localization, where each node can only measure distances, bearings, or relative positions in its local coordinate frame \cite{lin2017mix}. 
Compared with his result \cite{lin2017mix}, the advantages of our work are concluded as follows.
\begin{enumerate}[(i)]
    \item The work in \cite{lin2017mix} only considers 2-D space. The 3-D case cannot be solved by trivially extending the results in \cite{lin2017mix}, but our method can be applied in 3-D space.
    \item The work in \cite{lin2017mix} only considers three kinds of local relative measurements: local relative position, local relative bearing, and distance, but our work considers five kinds of local relative measurements. 
    \item The difference between the Assumption in  \cite{lin2017mix} and our Assumption \ref{ass4} is that the work in \cite{lin2017mix} requires the convex hull assumption, i.e., each node lies strictly inside the convex hull spanned by its neighboring nodes, but our work only requires the free node and its neighbors be non-colinear. It is clear that our Assumption \ref{ass4} is mild.
\end{enumerate}

\begin{figure}[t]
\centering
\includegraphics[width=1\linewidth]{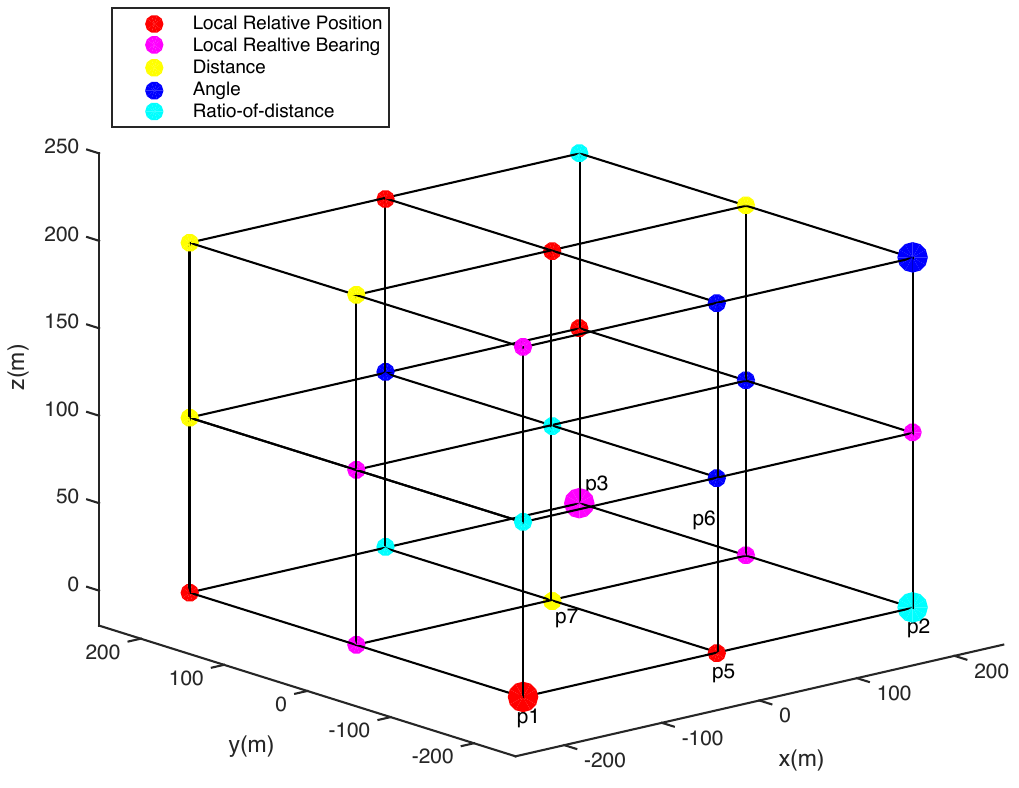}
\caption{3-D localizable network. }
\label{t1a}
\end{figure}

\section{Simulation}\label{simu}

Two numerical examples are given to illustrate the effectiveness of the theoretical findings. A three-dimensional localizable network with mixed types of measurements is shown in Fig. \ref{t1a}, where the $4$ anchor nodes $p_a = [p_1^T,\cdots ,p_{4}^T]^T$ and $23$ free nodes $p_f = [p_5^T,\cdots ,p_{27}^T]^T$
are denoted by large solid dots and small solid dots, respectively. The orientation difference between the local coordinate frame of each node and global coordinate frame is unknown.
It can be verified that the network in Fig. \ref{t1a} is infinitesimally angle-displacement rigid, where
the corresponding angle constraints and displacement constraints are constructed by our proposed methods shown in Section \ref{conc}. For example, for the
anchor nodes $p_1 \in \mathcal{R}$, $p_2 \in \mathcal{R}o\mathcal{D}$, $p_3 \in  \mathcal{B}$ and free nodes $p_5 \in \mathcal{R}$, $p_6 \in \mathcal{A}, p_7 \in \mathcal{D}$  shown in Fig. \ref{t1a},  the angle constraints \eqref{sa} for the anchor nodes are obtained by \eqref{para}.
\begin{equation}\label{sa}
 e_{13}^Te_{12}-e_{31}^Te_{32}=0, \ \ e_{21}^Te_{23}=0,
\end{equation}
where
\begin{equation}
     p_1 \!=\! [-200,-200,0]^T, p_2 \!=\! [200,-200,0]^T,  p_3 \!=\! [200,200,0]^T.
\end{equation}

For the free node $p_5$, since it can measure local relative position $e_{51}^5, e_{52}^5, e_{56}^5, e_{57}^5$ to its neighboring nodes $p_1, p_2, p_6, p_7$, the displacement constraint \eqref{sd} among the nodes $p_1, p_2, p_5, p_6, p_7$ are obtained by \eqref{wmi}.
\begin{equation}\label{sd}
    e_{51} +e_{52} =0,
\end{equation}
where
\begin{equation}
  p_5 \!=\! [0,-200,0]^T,  p_6 \!=\! [0,-200,100]^T,  p_7 \!=\! [0,0,100]^T.
\end{equation}

\begin{figure}[t]
\centering
\includegraphics[width=1\linewidth]{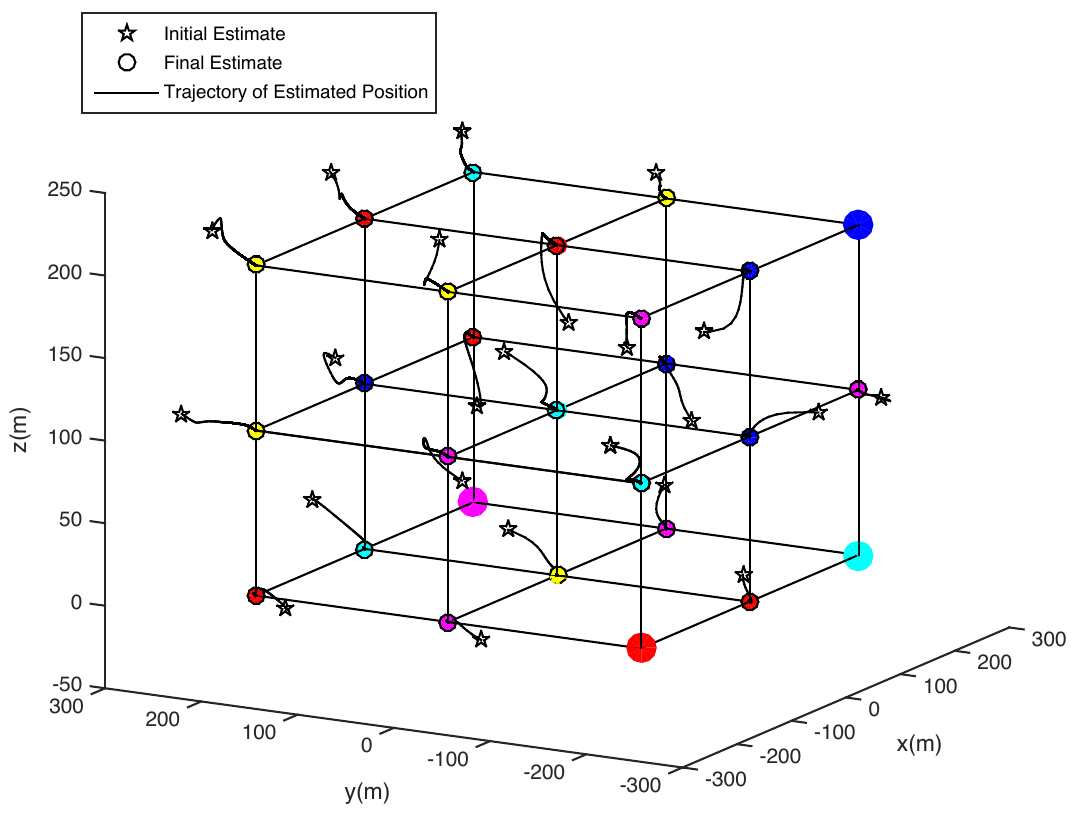}
\caption{Trajectory of estimated position without measurement noise. }
\label{case8}
\end{figure}

\begin{figure}[t]
\centering
\includegraphics[width=1\linewidth]{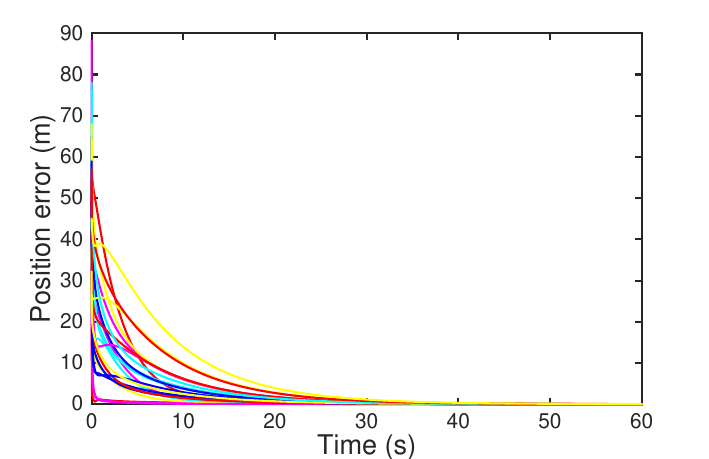}
\caption{Position estimation error without measurement noise. }
\label{case81}
\end{figure}

When there is no measurement noise, 
under
distributed localization algorithm
\eqref{law} with randomly generated initial estimate, each free node can estimate its position shown in Fig. \ref{case8}. The position estimation errors of the free nodes converge to zero globally shown in Fig. \ref{case81}.
In the proposed method, the mixed types of measurements are only used to construct the displacement constraints. Hence, the measurement noises will only influence displacement constraints.  For example, when there is measurement noise, the free node $p_5$ obtain the noisy local relative position measurements $\bar e_{51}^5, \bar e_{52}^5, \bar e_{56}^5, \bar e_{57}^5$ shown as
\begin{equation}
\begin{array}{ll}
     & \bar e_{51}^5 = e_{51}^5-[5,5,5]^T, \ \  \bar e_{52}^5 = e_{52}^5 +[6,7,8]^T, \\
     & \bar e_{56}^5= e_{56}^5-[3,4,4]^T, \ \ \bar e_{57}^5 = e_{57}^5 +[10,8,9]^T.
\end{array}
\end{equation}

Then, the displacement constraint with noisy local relative position measurements $\bar e_{51}^5, \bar e_{52}^5, \bar e_{56}^5, \bar e_{57}^5$ among the nodes $p_1, p_2, p_5, p_6, p_7$ becomes
\begin{equation}\label{sd1}
    \frac{6201}{61} e_{51}+\frac{38449}{380} e_{52}-\frac{1103}{551} e_{56} - e_{57} = 0.
\end{equation}

In the simulation, zero-mean Gaussian noises 
are added to local relative measurements. 
Based on the noisy local relative measurements, we can obtain the noisy matrices $\bar M_{ff}$ and $\bar M_{fa}$ in \eqref{noi}. The distributed algorithm \eqref{dis} becomes
\begin{equation}
\dot  {\hat{p}}_f = -\bar M_{ff}{\hat{p}}_f - \bar M_{fa}p_a.
\end{equation}

The trajectories of estimated positions and position estimation errors are given in Fig. \ref{case811} and Fig. \ref{case812}, respectively. The estimates of the free nodes can still converge to neighborhoods of the actual positions of the nodes shown if the condition \eqref{perror} holds.

\begin{figure}[t]
\centering
\includegraphics[width=1\linewidth]{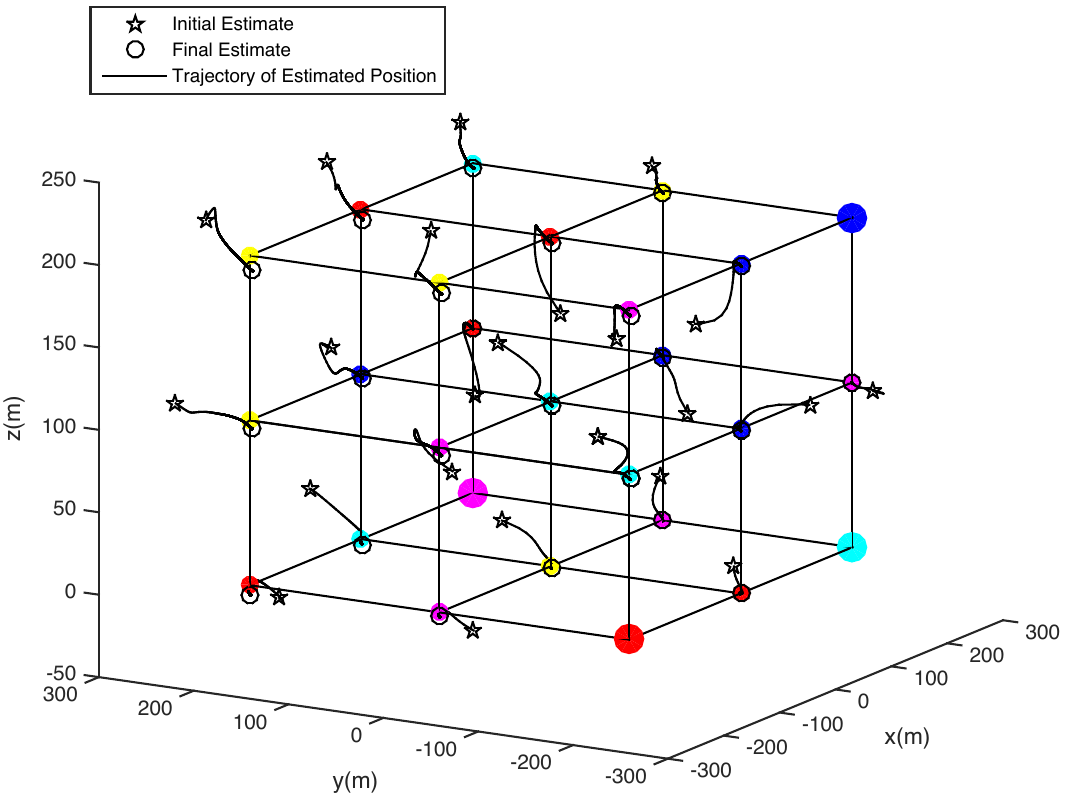}
\caption{Trajectory of estimated position with measurement noises. }
\label{case811}
\end{figure}

\begin{figure}[t]
\centering
\includegraphics[width=1\linewidth]{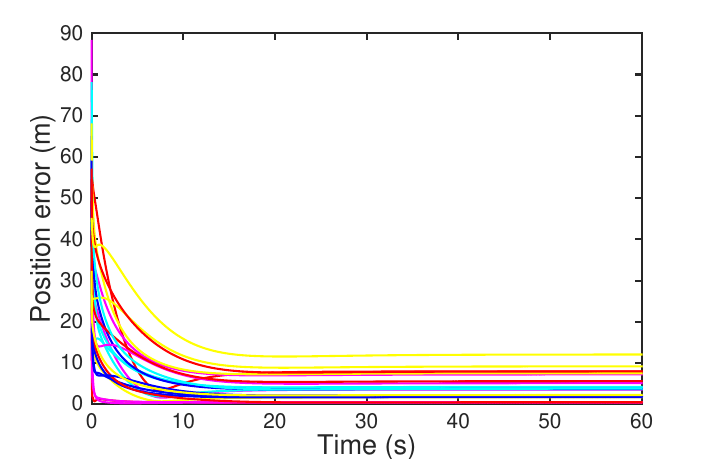}
\caption{Position estimation error with measurement noises. }
\label{case812}
\end{figure}

\section{Conclusion}\label{conc1}

This paper studies the 3-D network localization with mixed local relative measurements, where each node holds a local coordinate frame without a common orientation and has no knowledge about the global coordinate frame.
The main idea is to construct the displacement constraints for the positions of the nodes by using mixed local relative measurements. A linear distributed algorithm is proposed for each free node that can globally estimate its own position if the network is localizable. Our future work is to relax the assumption and evaluate the estimation error given the statistics of measurements noises.

\ifCLASSOPTIONcaptionsoff
  \newpage
\fi

\bibliographystyle{IEEEtran}
\bibliography{papers}

\end{document}